\documentclass{llncs}

\newif\iflongVersion
\longVersiontrue

\usepackage{footmisc}

\usepackage{enumitem}
\usepackage{fancyvrb}
\usepackage{color}
\usepackage{stmaryrd}
\usepackage{amsmath}
\usepackage{amsfonts}
\usepackage[all]{xy}
\usepackage{graphicx}

\usepackage{hyperref}  

\usepackage{letltxmacro} 

\usepackage[numbers]{natbib}
\usepackage{thmtools}
\usepackage{thm-restate}
\usepackage{etoolbox}

\declaretheoremstyle[
]{mystyle}

\makeatletter
\patchcmd{\@counteralias}
 {\@ifdefinable{c@#1}}
 {\expandafter\@ifdefinable\csname c@#1\endcsname}
 {}{}
\makeatother

\numberwithin{equation}{section}

\declaretheorem[name=Theorem,sharenumber=theorem,style=mystyle]{thm}
\declaretheorem[name=Lemma,sharenumber=lemma,title=Lemma,style=mystyle]{lem}
\declaretheorem[name=Proposition,sharenumber=proposition,title=Proposition,style=mystyle]{prop}
\declaretheorem[name=Lemma,title=Lemma,sharenumber=lemma,style=mystyle]{nlem}

\usepackage{listings}

\lstnewenvironment{code}
{
  \lstset{
  }
  \csname lst@SetFirstLabel\endcsname
}
{
  \csname lst@SaveFirstLabel\endcsname
}
\newcommand*{\SavedLstInline}{}
\LetLtxMacro\SavedLstInline\lstinline
\DeclareRobustCommand*{\lstinline}{%
  \ifmmode
    \let\SavedBGroup\bgroup
    \def\bgroup{%
      \let\bgroup\SavedBGroup
      \hbox\bgroup
    }%
  \fi
  \SavedLstInline
}

\lstnewenvironment{edoc}
{
  \lstset{
  }
  \csname lst@SetFirstLabel\endcsname
}
{
  \csname lst@SaveFirstLabel\endcsname
}

\usepackage[T1]{fontenc}
\usepackage{inconsolata}

\usepackage{xcolor}
\newcommand{\dnote}[1]{}
\newcommand{\lnote}[1]{}
\newcommand{\mnote}[1]{}

\newcommand{\ie}{i.e.}
\newcommand{\eg}{e.g.}

\let\olddefinition\definition
\renewcommand{\definition}{\olddefinition\normalfont}
\let\oldcorollary\corollary
\renewcommand{\corollary}{\oldcorollary\normalfont}
\let\oldlemma\lemma
\renewcommand{\lemma}{\oldlemma\normalfont}
\let\oldproposition\proposition
\renewcommand{\proposition}{\oldproposition\normalfont}
\let\oldtheorem\theorem
\renewcommand{\theorem}{\oldtheorem\normalfont}
\let\oldprop\prop
\renewcommand{\prop}{\oldprop\normalfont}
\let\oldthm\thm
\renewcommand{\thm}{\oldthm\normalfont}
\let\oldlem\lem
\renewcommand{\lem}{\oldlem\normalfont}

\iflongVersion
\renewcommand{\eqref}[1]{(\ref{#1})}
\fi

\title{Composing bidirectional programs monadically\iflongVersion*\fi}
\author{Li-yao Xia\inst{1} \and Dominic Orchard\inst{2} \and Meng Wang\inst{3}}
\institute{University of Pennsylvania \and University of Kent \and
  University of Bristol}

\begin{document}
\maketitle

\iflongVersion
\vspace{-1.25em}
\fi

\begin{abstract}
  Software frequently converts data from one representation to another
and vice versa. Na\"{i}vely specifying both conversion directions separately
is error prone and introduces conceptual duplication. Instead, \emph{bidirectional
programming} techniques allow programs to be written which can be interpreted in
both directions. However, these techniques often employ unfamiliar programming idioms
via restricted, specialised combinator libraries.
Instead, we introduce a framework for composing bidirectional programs monadically,
enabling bidirectional programming with familiar abstractions in functional languages such as Haskell.
We demonstrate the generality of our approach applied to
parsers/printers, lenses, and generators/predicates. We show how
to leverage compositionality and equational reasoning for the
verification of \emph{round-tripping properties} for such monadic bidirectional
programs.

\end{abstract}

\iflongVersion
  \vspace{-2.5em}
  \renewcommand{\thefootnote}{*}
  \footnotetext{Published in the proceedings of ESOP 2019. This version
    includes the appendices.}
  \renewcommand{\thefootnote}{\arabic{footnote}}
  \addtocounter{footnote}{-1}
\fi

\section{Introduction}
\label{sec:introduction}
A \emph{bidirectional transformation} (BX) is a pair of mutually
related mappings between source and target data objects.  A well-known
example solves the \emph{view-update
  problem}~\citep{BaSp81} from relational database design. A
\emph{view} is a derived database table, computed from concrete
\emph{source} tables by a query. The problem is to map an update of
the view back to a corresponding update on the source tables. This is
captured by a bidirectional transformation.  The bidirectional pattern
is found in a broad range of applications, including
parsing~\citep{ReOs10,wang2013flippr}, refactoring~\citep{Schuster2016},
code generation~\citep{Pombrio2014, Mayer2018}, and model
transformation~\citep{Stevens2008} and XML transformation~\citep{Pacheco2014}.

When programming a bidirectional transformation, one can separately
construct the forwards and backwards functions.
However, this approach duplicates effort, is prone to error, and
causes subsequent maintenance issues.  These problems can be avoided
by using specialised programming languages that generate
both directions from a single
definition~\citep{FGMPS07,Voigtlander09bff,MaHNHT07}, a discipline
known as \emph{bidirectional programming}.

The most well-known language family for BX programming is \emph{lenses}~\citep{FGMPS07}. A lens
captures transformations between sources $S$ and views $V$ via a pair
of functions \verb|get| $: S \to V$ and \verb|put| $: V \to  S \to  S$. The
\verb|get| function extracts a view from a
source and \verb|put| takes an updated view and a
source as inputs to produce an updated source. The asymmetrical
nature of \verb|get| and \verb|put| makes it possible for \verb|put| to recover some of
the source data that is not present in the view. In other words,
\verb|get| does not have to be injective to have a corresponding~\verb|put|.

Bidirectional transformations typically respect
\emph{round-tripping} laws, capturing the extent to which the
transformations preserve information between the two data
representations. For example, \emph{well-behaved lenses}~\citep{bohannon2006relational,FGMPS07}
should satisfy:
\begin{equation*}
\texttt{put} \;  (\texttt{get} \; s)\; s  = s \quad\quad\quad\quad
\texttt{get} \; (\texttt{put} \; v \; s) = v
\end{equation*}
Lens languages are typically designed to enforce these
properties.
This focus on unconditional correctness inevitably leads to reduced
practicality in programming:
lens combinators
are often stylised and disconnected from established programming
idioms. In this paper, we instead focus on expressing
bidirectional programs directly, using
 monads as an interface for sequential composition.

Monads are a
 popular pattern~\citep{wadler1995monads} (especially in Haskell)
which
combinator libraries in other
domains routinely exploit.
Introducing monadic composition to BX programming
significantly expands the expressiveness of BX languages and opens up
a route for programmers to explore the connection between BX programming and mainstream uni-directional programming.
Moreover, it appears that many applications of
bidirectional transformations (e.g., parsers and printers~\citep{wang2013flippr}) do not
share the lens \textit{get}/\textit{put} pattern, and as a result have
not been sufficiently explored. However, monadic
composition is known to be an effective way to construct at least one direction
of such transformations (e.g., parsers).

\paragraph{Contributions}
In this paper, we deliberately avoid the well-tried
approach of specialised lens languages, instead exploring a novel point in the BX design space based
on monadic programming, naturally reusing host language
constructs. We revisit lenses, and two more bidirectional patterns,
demonstrating how they can be subject to monadic programming.
By being uncompromising about
the monad interface, we expose the essential ideas behind our
framework whilst maximising its utility.
The trade off with our approach is that we can no longer enforce
correctness in the same way as conventional lenses:
our interface does not rule out all non-round-tripping BXs.
We tackle this issue by proposing a new compositional
reasoning framework that is flexible enough to characterise
a variety of round-tripping properties, and simplifies the necessary reasoning.

Specifically, we make the following contributions:
%
%
\begin{itemize}[leftmargin=1.1em]
  \item We describe a method to enable \emph{monadic
    composition} for bidirectional programs (Section~\ref{sec:monadic-profunctors}).
    Our approach is based on a construction which generates a
    \emph{monadic profunctor}, parameterised by
    two application-specific monads which are used
    to generate the \emph{forward} and \emph{backward} directions.

  \item To demonstrate the flexibility of our approach, we apply the above method to
    three different problem domains: parsers/printers
    (Section~\ref{sec:monadic-profunctors} and \ref{sec:compositionality}),
    lenses (Section~\ref{sec:lenses}), and generators/predicates for
    structured data (Section~\ref{sec:generators}).
    While the first two are well-explored areas in the bidirectional
    programming literature, the third one is a completely new
    application domain.

  \item We present a scalable reasoning framework, capturing notions
    of \emph{compositionality} for bidirectional properties
    (Section~\ref{sec:compositionality}).  We define classes of
    round-tripping properties inherent to bidirectionalism, which can
    be verified by following simple criteria. These principles are
    demonstrated with our three examples.
    We include some proofs
    for illustration in the paper. The supplementary material~\cite{lib} contains machine-checked Coq proofs for the main theorems. 

  \item We have implemented these ideas as Haskell libraries~\citep{lib}, %
    with two wrappers around \textsf{attoparsec} for parsers and printers,
    and \textsf{QuickCheck} for generators and predicates, showing the
    viability of our approach for real programs.
\end{itemize}

\noindent
We use Haskell for concrete examples, but the programming
patterns can be easily expressed in many functional
languages. We use the Haskell notation of assigning type signatures to expressions
via an infix double colon ``\lstinline{::}''.


\subsection{Further Examples of BX}
\label{sec:background}
We introduced lenses briefly above. We
now introduce the other two examples used in this paper: \emph{parsers/printers}
and \emph{generators/predicates}.

\paragraph{Parsing and printing}
Programming language tools (such as interpreters,
compilers, and refactoring tools) typically require two intimately linked
components: \emph{parsers} and \emph{printers}, respectively
mapping from source code to ASTs and back. A simple implementation of these
two functions can be given with types:
\begin{edoc}
      parser :: String -> AST           printer :: AST -> String
\end{edoc}
%
Parsers and printers are rarely actual inverses to each other, but instead
typically exhibit a variant of round-tripping such as:
\begin{equation*}
\lstinline{parser} \,\circ\, \lstinline{printer} \,\circ\, \lstinline{parser} \,\equiv\,
\lstinline{parser} \qquad
\lstinline{printer} \,\circ\, \lstinline{parser} \,\circ\, \lstinline{printer} \,\equiv\, \lstinline{printer}
\end{equation*}
The left equation describes the common situation that
parsing discards information about source code, such as
whitespace, so that printing the resulting AST does not recover the
original source. However, printing retains enough information
such that parsing the printed output yields an AST which is equivalent to
the AST from parsing the original source. The right equation
describes the dual: printing may map different ASTs to the
same string. For example, printed code $1+2+3$ might be produced
by left- and right-associated syntax trees.

For particular AST subsets, printing and parsing may actually be left- or right- inverses to
each other. 
Other characterisations are also
possible, e.g., with equivalence classes of ASTs (accounting for
reassociations).
%
Alternatively, parsers and printers may satisfy properties
about the interaction of partially-parsed inputs with the printer
and parser, \eg{}, if \lstinline{parser :: String -> (AST, String)}:
\begin{code}
(let (x, s') = parser s in parser ((printer x) ++ s'))  equiv  parser s
\end{code}
Thus, parsing and printing follows a pattern
of inverse-like functions which does not fit
the lens paradigm. The pattern resembles lenses between a source
(source code) and view (ASTs), but with a compositional notion for
the source and partial ``gets'' which consume some of the source,
leaving a remainder. \\[-0.75em]

Writing parsers and printers by hand is often tedious due to the redundancy
implied by their inverse-like relation.
Thus, various approaches have been
proposed for reducing the effort of writing parsers/printers by generating
both from a common definition~\cite{ReOs10,wang2013flippr, Matsuda2018}.

\paragraph{Generating and checking}
Property-based testing (\eg{}, QuickCheck)~\citep{Claessen2000}
expresses program properties as executable predicates. For instance,
the following property checks that an insertion function \lstinline{insert}, given a
sorted list --- as checked by the predicate
\lstinline{isSorted :: [Int] -> Bool} --- produces another sorted list.
The combinator \lstinline{==>} represents implication for properties.
\begin{code}
propInsert :: Int -> [Int] -> Property
propInsert val list = isSorted list ==> isSorted (insert val list)
\end{code}
To test this property, a testing framework generates random inputs for
\lstinline{val} and
\lstinline{list}. The implementation of \lstinline{==>} applied here
first checks whether \lstinline{list} is sorted, and if it is,
checks that \lstinline{insert val list} is sorted as well. This process is repeated
with further random inputs until either a counterexample is found or a predetermined number of test cases
pass.

However, this na\"{i}ve method is inefficient: many properties
such as \lstinline{propInsert} have preconditions which are satisfied by an
extremely small fraction of inputs. In this case, the ratio of sorted lists
among lists of length $n$ is inversely proportional to $n!$, so most generated
inputs will be discarded for not satisfying the \lstinline{isSorted} precondition.
Such tests give no information about the validity of the predicate being tested
and thus are prohibitively inefficient.

When too many inputs are being discarded, the user must instead supply the
framework with \emph{custom generators} of values satisfying the precondition:
\lstinline{genSorted :: Gen [Int]}.

One can expect two complementary properties of such a generator.
A generator is \emph{sound} with respect to the predicate \lstinline{isSorted}
if it generates only values satisfying \lstinline{isSorted}; soundness means that no
tests are discarded, hence the tested property is better exercised.
A generator is \emph{complete} with respect to \lstinline{isSorted} if it can
generate all satisfying values; completeness ensures the correctness of testing
a property with \lstinline{isSorted} as a precondition, in the sense that if there
is a counterexample, it will be eventually generated. In this setting of testing,
completeness, which affects the potential adequacy of testing, is arguably more important than soundness, which affects only efficiency.

It is clear that generators and predicates are
closely related, forming a pattern similar to that of
bidirectional transformations. Given that good
generators are usually difficult to construct,
the ability to extract both from a common specification
with bidirectional programming is a very attractive
alternative.

\paragraph{Roadmap}
We begin by outlining a concrete example
of our monadic approach via parsers and printers (Section~\ref{sec:biparsers}), before explaining the general approach of using \emph{monadic profunctors} to structure bidirectional programs (Section~\ref{sec:monadic-profunctors}).
Section~\ref{sec:compositionality} then presents a
compositional reasoning framework for monadic bidirectional programs, with
varying degrees of strength adapted to different round-tripping properties.
We then replay the developments of the earlier sections to define lenses
as well as generators and predicates in Sections~\ref{sec:lenses}
and~\ref{sec:generators}.


\section{Monadic bidirectional programming}
\label{sec:biparsers}
%

A bidirectional parser, or \emph{biparser}, combines both a parsing
direction and printing direction. Our first novelty here
is to express biparsers monadically.

In code samples, we use the Haskell pun of naming
variables after their types, e.g., a variable of some abstract type
\lstinline{v} will also be called \lstinline{v}. Similarly, for some
type constructor \lstinline{m}, a variable of type \lstinline{m v} will
be called \lstinline{mv}. A function \lstinline{u -> m v} (a
Kleisli arrow for a monad \lstinline{m}) will be called
\lstinline{kv}.

\paragraph{Monadic parsers}

The following data type provides the standard way to describe parsers of values
of type \lstinline{v} which may consume only part of
the input string:
\begin{edoc}
data Parser v = Parser { parse :: String -> (v, String) }
\end{edoc}
It is well-known that such parsers are
monadic~\citep{wadler1995monads}, 
i.e., they have a notion of monadic sequential
composition embodied by the interface:
\begin{edoc}
instance Monad Parser where
  (>>=) :: Parser v -> (v -> Parser w) -> Parser w
  return :: v -> Parser v
\end{edoc}
The sequential composition operator \lstinline{(>>=)}, called \emph{bind},
describes the scheme of constructing a parser by sequentially
composing two sub-parsers where the second depends on the output of
the first; a parser of \lstinline{w} values is made up of a parser of \lstinline{v}
and a parser of \lstinline{w} that depends on the previously
parsed \lstinline{v}. Indeed, this is the implementation
given to the monadic interface:
\begin{edoc}
  pv >>= kw = Parser (\s -> let (v, s') = parse pv s in parse (kw v) s')
  return v  = Parser (\s -> (v, s))
\end{edoc}
Bind first runs the parser \lstinline{pv} on an input
string \lstinline{s}, resulting in a value \lstinline{v} which is used to create the
parser \lstinline{kw v}, which is in turn run on the remaining input
\lstinline{s'} to produce parsed values of type \lstinline{w}.
The \lstinline{return} operation creates a trivial parser for
any value \lstinline{v} which does not consume any input but simply
produces \lstinline{v}.

In practice, parsers composed with \lstinline{(>>=)} often have a relationship
between the output types of the two operands: usually that the former ``contains''
the latter in some sense.
For example, we might parse an expression and compose this
with a parser for statements, where statements contain
expressions. This relationship will be useful later when we consider
printers.

As a shorthand, we can discard the remaining unparsed string of a parser using projection, giving
a helper function \lstinline{parser :: Parser v -> (String -> v)}.

\paragraph{Monadic printers}

Our goal is to augment parsers with their inverse printer, such that
we have a monadic type \lstinline{Biparser} which provides two
complementary (bi-directional) transformations:
\begin{edoc}
  parser   :: Biparser v -> (String -> v)
  printer  :: Biparser v -> (v -> String)
\end{edoc}
However, this type of printer \lstinline{v -> String} (shown also in Section~\ref{sec:background}) cannot
form a monad because it is \emph{contravariant} in its type
parameter \lstinline{v}. Concretely, we cannot implement the
 bind (\lstinline{>>=}) operator for values with
types of this form:
\begin{edoc}
-- Failed attempt
bind :: (v -> String) -> (v -> (w -> String)) -> (w -> String)
bind pv kw = \w -> let v = (??) in pv v ++ kw v w
\end{edoc}
We are stuck trying to fill the hole (\lstinline{??}) as there is no
way to get a value of type \lstinline{v} to pass as an argument
to \lstinline{pv} (first printer) and \lstinline{kw} (second
printer which depends on a \lstinline{v}). Subsequently,
we cannot construct a monadic biparser by simply taking a product of the
parser monad and \lstinline{v -> String} and leveraging the result
that the product of two monads is a monad.

But what if the type variables of \lstinline{bind} were related by containment, such that
\lstinline{v} is contained within \lstinline{w} and thus we have a projection
\lstinline{w -> v}? We could use this projection to fill the hole in the failed
attempt above, defining a bind-like operator:
\begin{code}
bind' :: (w -> v) -> (v -> String) -> (v -> (w -> String)) -> (w -> String)
bind' from pv kw = \w -> let v = from w in pv v ++ kw v w
\end{code}
This is closer to the monadic form, where \lstinline{from :: w -> v}
resolves the difficulty of contravariance by ``contextualizing'' the printers.
Thus, the first printer is no longer just ``a
printer of \lstinline{v}'', but ``a printer of \lstinline{v}
extracted \lstinline{from} \lstinline{w}''.
In the context of constructing a bidirectional
parser, having such a function to hand is not an unrealistic expectation: recall that
when we compose two parsers, typically the values of the first parser
for \lstinline{v} are contained within the values returned by the second
parser for \lstinline{w}, thus a notion of projection can be defined
and used here to recover a \lstinline{v} in order to build the corresponding printer compositionally.

Of course, this
is still not a monad. However, it suggests a way to
generate a monadic form by putting the printer and the contextualizing projection together,
\lstinline{(w -> v, v -> String)} and fusing them into \lstinline{(w -> (v,  String))}.
This has the advantage of removing the contravariant occurence of
\lstinline{v}, yielding a data type:
\begin{edoc}
data Printer w v = Printer { print :: w -> (v, String) }
\end{edoc}
If we fix the first parameter type \lstinline{w}, then the type
\lstinline{Printer w} of printers for \lstinline{w} values is indeed
monadic, combining a \emph{reader monad} (for some global
read-only parameter of type \lstinline{w}) and a \emph{writer monad} (for
strings), with implementation:
%
\begin{edoc}
instance Monad (Printer w) where
  return :: v -> Printer w v
  return = \v -> Printer (\_ -> (v, ""))

  (>>=) :: Printer w v -> (v -> Printer w t) -> Printer w t
  pv >>= kt = Printer (\w -> let (v, s) = print pv w
                                 (t, s') = print (kt v) w in (t, s ++ s'))
\end{edoc}
The printer \lstinline{return v} ignores its input and prints nothing.
For bind, an input \lstinline{w} is shared by both
printers and the resulting strings are concatenated.

We can adapt the contextualisation of a printer by the following
operation which amounts to pre-composition, witnessing the
fact that \lstinline{Printer} is a contravariant functor in its first parameter:
\begin{edoc}
  comap :: (w -> w') -> Printer w' v -> Printer w v
  comap from (Printer f) = Printer (f . from)
\end{edoc}

\subsection{Monadic biparsers}
\label{subsec:biparsers}

So far so good: we now have a monadic notion of printers. However, our
goal is to combine parsers and
printers in a single type. Since we have two monads, we use the standard result
that a product of monads is a monad, defining \emph{biparsers}:
\begin{code}
data Biparser u v = Biparser { parse :: String -> (v, String)
                             , print :: u      -> (v, String) }
\end{code}
By pairing parsers and printers we have to unify their
covariant parameters. When both the type parameters of \lstinline{Biparser} are the same it
is easy to interpret this type: a biparser \lstinline{Biparser v v}
is a parser from strings to \lstinline{v} values and printer from
\lstinline{v} values to strings. We refer to biparsers of this
type as \emph{aligned} biparsers.
What about when the type parameters differ? A biparser of
type \lstinline{Biparser u v} provides a parser from strings
to \lstinline{v} values and a printer from \lstinline{u}
values to strings, but where the printers can compute
\lstinline{v} values from \lstinline{u} values, i.e.,
\lstinline{u} is some common broader representation which contains
relevant \lstinline{v}-typed subcomponents. A biparser
\lstinline{Biparser u v} can be thought of as printing a certain
subtree \lstinline{v} from the broader representation of a syntax tree
\lstinline{u}.

The corresponding monad for \lstinline{Biparser} is the product of the
previous two monad definitions for \lstinline{Parser} and
\lstinline{Printer}, allowing both to be composed sequentially at the
same time.  To avoid duplication we elide the definition here which is
shown in full in
\iflongVersion
Appendix~\ref{app:further-code}.
\fi

We can also lift the previous notion of \lstinline{comap}
from printers to biparsers, which gives us a way to contextualize a printer:
\begin{code}
comap :: (u -> u') -> Biparser u' v -> Biparser u v
comap f (Biparser parse print) = Biparser parse (print . f)

upon :: Biparser u' v -> (u -> u') -> Biparser u v
upon = flip comap
\end{code}
%
%
In the rest of this section, we use the
alias ``\lstinline{upon}'' for \lstinline{comap} with flipped parameters
where we read \lstinline{p `upon` subpart} as applying the
printer of \lstinline{p :: Biparser u' v} on a subpart of an
input of type \lstinline{u}
calculated by \lstinline{subpart :: u -> u'}, thus yielding a biparser
of type \lstinline{Biparser u v}.

\paragraph{An example biparser}

Let us write a biparser, \lstinline{string :: Biparser String String}, for strings
which are prefixed by their length and a space. For example, the following unit tests should be true:
\begin{code}
test1 = parse string "6 lambda calculus" == ("lambda", " calculus")
test2 = print string "SKI" == ("SKI", "3 SKI")
\end{code}
We start by defining a primitive biparser of single characters as:
\begin{code}
char :: Biparser Char Char
char = Biparser (\ (c : s) -> (c, s)) (\ c -> (c, [c]))
\end{code}
A character is parsed by deconstructing the source string into its
head and tail. For brevity, we do not handle the failure associated
with an empty string. A character \lstinline{c} is printed as its
single-letter string (a singleton list) paired with \lstinline{c}.

Next, we define a biparser \lstinline{int} for an
integer followed by a single space. An auxiliary biparser
\lstinline{digits} (on the right) parses an integer one digit at a time
into a string.  Note that in Haskell, the \lstinline{do}-notation
statement ``\lstinline{d <- char `upon` head}'' desugars to
``\lstinline{char `upon` head >>= \ d ->} \ldots'' which uses
\lstinline{(>>=)} and a function binding \lstinline{d} in the scope of
the rest of the desugared block.

\vspace{-1.5em}
\begin{center}
\hspace{-2.5em}
\begin{minipage}[t]{0.46\linewidth}
\begin{code}
int :: Biparser Int Int
int = do
  ds <- digits `upon` printedInt
  return (read ds)
  where
    printedInt n = show n ++ " "


\end{code}
\end{minipage}
\begin{minipage}[t]{0.50\linewidth}
\begin{code}
digits :: Biparser String String
digits = do
  d <- char `upon` head
  if isDigit d then do
    igits <- digits `upon` tail
    return (d : igits)
  else if d == ' ' then return " "
  else error "Expected digit or space"
\end{code}
\end{minipage}
\vspace{-0.5em}
\end{center}

\noindent
On the right, \lstinline{digits} extracts a \lstinline{String} consisting of digits
followed by a single space.
As a parser, it parses a character (\lstinline{char `upon` head}); if it is a digit then
it continues parsing recursively (\lstinline{digits `upon` tail}) appending the first
character to the result (\lstinline{d : igits}). Otherwise, if the
parsed character is a space the parser returns \lstinline{" "}. As a printer, \lstinline{digits} expects
a non-empty string of the same format; \lstinline{`upon` head} extracts the
first character of the input, then \lstinline{char} prints it and returns it back as
\lstinline{d}; if it is a digit, then \lstinline{`upon` tail} extracts the rest
of the input to print recursively. If the character is a space,
the printer returns a space and terminates; otherwise (not digit or space) the printer throws an error.

On the left, the biparser \lstinline{int} uses \lstinline{read} to
convert an input string of digits (parsed by \lstinline{digits}) into an
integer, and \lstinline{printedInt} to convert an integer to an output
string printed by \lstinline{digits}.
A safer implementation could
return the \lstinline{Maybe} type when parsing but we keep things
simple here for now.

After parsing an integer \lstinline{n}, we can parse the string following it
by iterating \lstinline{n} times the biparser \lstinline{char}. This is captured by
the \lstinline{replicateBiparser} combinator below, defined recursively like \lstinline{digits} but
with the termination condition given by an external parameter.
To iterate \lstinline{n} times a biparser \lstinline{pv}: if \lstinline{n == 0},
there is nothing to do and we return the empty list; otherwise for
\lstinline{n > 0}, we run \lstinline{pv} once to get the head \lstinline{v},
and recursively iterate \lstinline{n-1} times to get the tail \lstinline{vs}.

Note that although not reflected in its type, \lstinline{replicateBiparser n pv}
expects, as a printer, a list \lstinline{l} of length \lstinline{n}:
if \lstinline{n == 0}, there is nothing to print; if \lstinline{n > 0},
\lstinline{`upon` head} extracts the head of \lstinline{l} to print it with \lstinline{pv},
and \lstinline{`upon` tail} extracts its tail, of length \lstinline{n-1}, to print it
recursively.
\begin{code}
replicateBiparser :: Int -> Biparser u v -> Biparser [u] [v]
replicateBiparser 0 pv = return []
replicateBiparser n pv = do
  v  <- pv `upon` head
  vs <- (replicateBiparser (n - 1) pv) `upon` tail
  return (v : vs)
\end{code}
(akin to \lstinline{replicateM} from
Haskell's standard library). We can now fulfil our task:
\begin{code}
string :: Biparser String String
string = int `upon` length >>= \n -> replicateBiparser n char
\end{code}
Interestingly, if we erase applications of
\lstinline{upon}, \ie{}, we substitute every expression of the form
\lstinline{py `upon` f} with \lstinline{py} and ignore the second
parameter of the types, we obtain
what is essentially the definition of a parser in an idiomatic style
for monadic parsing. This is because
\lstinline{`upon` f} is the identity on the parser component of \lstinline{Biparser}.
Thus the biparser code closely resembles
standard, idiomatic monadic parser code but with ``annotations'' via
\lstinline{upon} expressing how to apply the backwards direction of printing to subparts of the
parsed string.


Despite its simplicity, the syntax of
length-prefixed strings is notably context-sensitive.
Thus the example makes crucial use of the monadic interface for
bidirectional programming:
a value (the length) must first be extracted
to dynamically delimit the string that is parsed next.
Context-sensitivity is standard for parser combinators
in contrast with parser generators, \eg{}, Yacc, and
applicative parsers, which are mostly restricted to context-free
languages. By our monadic BX approach, we can now bring this
power to bear on \emph{bidirectional} parsing.

\section{A unifying structure: monadic profunctors}
\label{sec:monadic-profunctors}

The biparser examples of the previous section were enabled by both the
monadic structure of \lstinline{Biparser} and the \lstinline{comap}
operation (also called \lstinline{upon}, with flipped arguments).  We
describe a type as being a \emph{monadic profunctor} when it has both
a monadic structure and a \lstinline{comap} operation (subject
to some equations). The notion of a monadic profunctor is general, but it
characterises a key class of structures for bidirectional programs,
which we explain here. Furthermore, we show a construction of monadic
profunctors from pairs of monads which elicits the necessary structure
for monadic bidirectional programming in the style of the previous
section.

\paragraph{Profunctors}

In Section~\ref{subsec:biparsers}, biparsers were defined by a data type
with two type parameters (\lstinline{Biparser u v})
which is functorial and monadic in the second parameter and
\emph{contravariantly} functorial in the first parameter (provided by
the \lstinline{comap} operation).
In standard terminology, a two-parameter type \lstinline{p} which is
functorial in both its type parameters is called a \emph{bifunctor}.
In Haskell, the term \emph{profunctor} has come to mean any bifunctor
which is contravariant in the first type parameter and covariant in
the
second.\footnote{{\url{http://hackage.haskell.org/package/profunctors/docs/Data-Profunctor.html}}}
This differs slightly from the standard category theory terminology
where a profunctor is a bifunctor
$\mathsf{F} : \mathcal{D}^{\mathsf{op}} \times \mathcal{C} \rightarrow
\textbf{Set}$. This corresponds to the Haskell community's use of the term
``profunctor'' if we treat Haskell in an idealised way as the category of sets.

We adopt this programming-oriented terminology, capturing
the \lstinline{comap} operation via a class \lstinline{Profunctor}.
In the preceding section, some uses of \lstinline{comap} involved a partial
function, \eg{}, \lstinline{comap head}. We make the possibility of
partiality explicit via the \lstinline{Maybe} type, yielding the following definition.

\begin{definition}\label{def:profunctor}
  A binary data type is a \textbf{profunctor} if it
  is a contravariant functor in its first parameter and covariant
  functor in its
  second, with the operation:
\begin{code}
class ForallF Functor p => Profunctor p where
     comap :: (u -> Maybe u') -> p u' v -> p u v
\end{code}
%
which should obey two laws:
\begin{align*}
  \lstinline{comap Just = id}
  \qquad\quad
  \lstinline{comap (f >=> g) = comap f . comap g}
\end{align*}
where \lstinline{(>=>) :: (a -> Maybe b) -> (b -> Maybe c) -> (a -> Maybe c)}
composes partial functions (left-to-right), captured by Kleisli arrows
of the \texttt{Maybe} monad.

The constraint \lstinline{ForallF Functor p} captures a universally
quantified constraint \citep{Bottu2017}: for all types \lstinline{u}
then \lstinline{p u} has an instance of the \lstinline{Functor}
class.%
\footnote{As of GHC 8.6, the \lstinline{QuantifiedConstraints} extension allows
  universal quantification in constraints, written as
  \texttt{forall u. Functor (p u)}, but for simplicity we use the constraint
  constructor \lstinline{ForallF} from the
  \textsf{constraints} package:
  \url{http://hackage.haskell.org/package/constraints}.}

The requirement for \lstinline{comap} to take partial
functions is in response to the frequent need to restrict the domain
of bidirectional transformations. In combinator-based approaches,
 combinators typically constrain bidirectional
programs to be bijections, enforcing domain restrictions
by construction. Our more flexible approach requires a way
to include such restrictions explicitly, hence \lstinline{comap}.

Since the contravariant part of the bifunctor applies to functions of type
\lstinline{u -> Maybe u'}, the categorical analogy here is more precisely
a profunctor $\mathsf{F} : {\mathcal{C}_T}^{\mathsf{op}} \times
\mathcal{C} \rightarrow \textbf{Set}$ where $\mathcal{C}_T$ is the
Kleisli category of the partiality (\lstinline{Maybe}) monad.
\end{definition}

\begin{definition}\label{def:promonad}
A \textbf{monadic profunctor}
is a profunctor \lstinline{p} (in the sense of Definition~\ref{def:profunctor})
such that \,\lstinline{p u}\, is a monad
for all \lstinline{u}.
In terms of type class constraints, this means there is
an instance \lstinline{Profunctor p} and for all \lstinline{u}
there is a \lstinline{Monad (p u)} instance. Thus, we represent
monadic profunctors
by the following empty class (which inherits all its methods from its superclasses):
\begin{code}
class (Profunctor p, ForallF Monad p) => Profmonad p
\end{code}
Monadic profunctors must obey the following laws about the interaction between
profunctor and monad operations:
\begin{edoc}
comap f (return y)     =   return y
comap f (py >>= kz)    =   comap f py >>= (\ y -> comap f (kz y))
\end{edoc}
%
(for all \lstinline{f :: u -> Maybe v},
\lstinline{py :: p v y}, \lstinline{kz :: y -> p v z}).
These laws are equivalent to saying that \lstinline{comap} lifts (partial)
functions into monad morphisms.
In Haskell, these laws are obtained \emph{for free}
by parametricity~\citep{Wadler89}.
This means that every contravariant functor and monad is in fact a
monadic profunctor, thus the following universal instance is lawful:
\begin{code}
instance (Profunctor p, ForallF Monad p) => Profmonad p
\end{code}
\end{definition}

\begin{corollary}
Biparsers form a monadic profunctor
as there is an instance of \lstinline{Monad (P u)} and \lstinline{Profunctor p}
 satisfying the requisite laws.
\end{corollary}


Lastly, we introduce a useful piece of terminology (mentioned in
the previous section on biparsers)
for describing values of a profunctor of a particular form:
\begin{definition}
\label{sec:aligned}
A value \lstinline{p :: P u v} of a profunctor \lstinline{P}
 is called \emph{aligned} if \lstinline{u} = \lstinline{v}.
\end{definition}

\subsection{Constructing monadic profunctors}
\label{sec:constructing-mp}

Our examples (parsers/printers, lenses, and generators/predicates) share monadic
profunctors as an abstraction, making it possible to write different kinds of
bidirectional transformations monadically. Underlying these definitions of
monadic profunctors is a common structure, which we explain here using biparsers, and which
will be replayed in Section~\ref{sec:lenses} for
lenses and Section~\ref{sec:generators} for bigenerators.

There are two simple ways in which a covariant functor \lstinline{m} (resp. a monad)
gives rise to a profunctor (resp. a monadic profunctor).
The first is by constructing a profunctor in which the contravariant
parameter is discarded, \ie{}, \lstinline{p u v = m v};
the second is as a function type from the contravariant
parameter \lstinline{u} to \lstinline{m v}, \ie{}, \lstinline{p u v = u -> m v}.
These are standard mathematical constructions, and the
latter appears in the Haskell \textsf{profunctors} package with the
name \lstinline{Star}. Our core construction is based on these two
ways of creating a profunctor, which we call \lstinline{Fwd} and \lstinline{Bwd} respectively:
\begin{code}
data Fwd m u v = Fwd { unFwd :: m v }      -- ignore contrv. parameter
data Bwd m u v = Bwd { unBwd :: u -> m v } -- maps from contrv. parameter
\end{code}

\noindent
The naming reflects the idea that these two
constructions will together capture a bidirectional transformation
and are related by domain-specific round-tripping properties in our
framework. Both \lstinline{Fwd} and \lstinline{Bwd} map any
functor into a profunctor by the following type class instances:

\begin{code}
instance Functor m => Functor (Fwd m u) where
  fmap f (Fwd x) = Fwd (fmap f x)
instance Functor m => Profunctor (Fwd m) where
  comap f (Fwd x) = Fwd x

instance Functor m => Functor (Bwd m u) where
  fmap f (Bwd x) = Bwd ((fmap f) . x)
instance (Monad m, MonadPartial m) => Profunctor (Bwd m) where
  comap f (Bwd x) = Bwd ((toFailure . f) >=> x)
\end{code}

\noindent There is an additional constraint here for \lstinline{Bwd}, enforcing
that the monad \lstinline{m} is a member of the \lstinline{MonadPartial}
class which we define as:
\begin{code}
class MonadPartial m  where  toFailure :: Maybe a -> m a
\end{code}
This provides an interface for monads which can internalise a notion
of failure, as captured at the top-level by \lstinline{Maybe} in \lstinline{comap}.

Furthermore, \lstinline{Fwd} and \lstinline{Bwd} both map any monad into
a monadic profunctor:

\vspace{-1em}
\hspace{-2.75em}
\begin{minipage}[t]{0.432\linewidth}
\begin{code}
instance Monad m
      => Monad (Fwd m u) where
  return x = Fwd (return x)
  Fwd py >>= kz =
    Fwd (py >>= unFwd . kz)
\end{code}
\end{minipage}
\begin{minipage}[t]{0.55\linewidth}
\begin{code}
instance Monad m
      => Monad (Bwd m u) where
  return x = Bwd (\_ -> return x)
  Bwd my >>= kz = Bwd
    (\u -> my u >>= (\y -> unBwd (kz y) u))
\end{code}
\end{minipage}
\vspace{0.5em}

\noindent
The product of two monadic profunctors is also a monadic profunctor.
This follows from the fact that the product of two monads is a monad and the
product of two contravariant functors is a contravariant functor.
\begin{code}
data (:*:) p q u v = (:*:) { pfst :: p u v,  psnd :: q u v }
\end{code}

\begin{code}
instance (Monad (p u), Monad (q u)) => Monad ((p :*: q) u) where
  return y = return y :*: return y
  py :*: qy >>= kz = (py >>= pfst . kz) :*: (qy >>= psnd . kz)

instance (ForallF Functor (p :*: q), Profunctor p, Profunctor q)
      => Profunctor (p :*: q) where
  comap f (py :*: qy) = comap f py :*: comap f qy
\end{code}
\vspace{-0.5em}

\noindent
\subsection{Deriving biparsers as monadic profunctor pairs}
\label{subsec:biparser-promonad}

We can redefine biparsers in terms of the above data types,
their instances, and two standard monads, the state and writer monads:
\begin{edoc}
type State  s a = s -> (a, s)
type WriterT w m a = m (a, w)
type Biparser = Fwd (State String) :*: Bwd (WriterT Maybe String)
\end{edoc}
The backward direction composes the writer monad with the
\lstinline{Maybe} monad using \lstinline{WriterT} (the writer monad
transformer, equivalent to composing two monads with a
distributive law). Thus the backwards component of \lstinline{Biparser}
corresponds to printers (which may fail) and the
forwards component to parsers:
\begin{align*}
  \lstinline{Bwd (WriterT Maybe String) u v} \quad & \cong
\quad \lstinline{u -> Maybe (v, String)} \\[-0.25em]
\lstinline{Fwd (State String) u v} \quad & \cong \quad \lstinline{String -> (v, String)}
\end{align*}
For the above code to work in Haskell, the
\lstinline{State} and \lstinline{WriterT} types need to be
defined via either a {\bf\lstinline{data}} type or
{\bf\lstinline{newtype}} in order to allow type
class instances on partially applied
  type constructors. We abuse
the notation here for simplicity but define smart
constructors and deconstructors for the actual implementation:\footnote{\emph{Smart constructors} (and dually
  \emph{smart deconstructors}) are just functions that
 hide boilerplate code for constructing and deconstructing data
 types.}
\begin{edoc}
parse :: Biparser u v -> (String -> (v, String))
print :: Biparser u v -> (u -> Maybe (v, String))
mkBP  :: (String -> (v, String)) -> (u -> Maybe (v, String)) -> Biparser u v
\end{edoc}
The monadic profunctor definition for biparsers now comes for free from the
constructions in Section~\ref{sec:constructing-mp} along with the
following instance of \lstinline{MonadPartial} for the writer monad
transformer with the \lstinline{Maybe} monad:
\begin{code}
instance Monoid w => MonadPartial (WriterT w Maybe) where
  toFailure Nothing = WriterT Nothing
  toFailure (Just a) = WriterT (Just (a, mempty))
\end{code}
In a similar manner, we will use this monadic profunctor construction to define monadic bidirectional
transformations for lenses (\S\ref{sec:lenses}) and bigenerators (\S\ref{sec:generators}).

The example biparsers from Section~\ref{subsec:biparsers} can be easily
redefined using the structure here. For example, the primitive
biparser \lstinline{char} becomes:
\begin{code}
char :: Biparser Char Char
char = mkBP (\ (c : s) -> (c, s)) (\ c -> Just (c, [c]))
\end{code}

\paragraph{Codec library}

The \textsf{codec} library~\citep{codec} provides a general type
for bidirectional programming isomorphic
to our composite type \lstinline{Fwd r :*: Bwd w}:
\begin{edoc}
data Codec r w c a = Codec { codecIn :: r a, codecOut :: c -> w a }
\end{edoc}
Though the original \textsf{codec} library was developed
independently, its current form is a result of this
work. In particular, we contributed to the package by generalising its original type (\lstinline{codecOut :: c -> w ()}) to the one above, and provided \lstinline{Monad} and
\lstinline{Profunctor} instances to support monadic bidirectional
programming with codecs.

\section{Reasoning about bidirectionality}
\label{sec:compositionality}
%

So far we have seen how the monadic profunctor structure provides a
way to define biparsers using familiar operations and syntax:
monads and \lstinline{do}-notation. This structuring allows
both the forwards and backwards components of a biparser to
be defined simultaneously in a single compact definition.

This section studies the interaction of monadic profunctors with the
\emph{round-tripping laws} that relate the two components of a
bidirectional program.  For every bidirectional transformation we can
define dual properties: \emph{backward round tripping} (going
backwards-then-forwards) and \emph{forward round tripping} (going
forwards-then-backwards). In each BX domain, such properties
also capture additional domain-specific information flow inherent to the
transformations.  We use biparsers as the running example.
We then apply the same principles to our other examples in
Sections~\ref{sec:lenses} and \ref{sec:generators}. For brevity, we
use \lstinline{Bp} as an alias for \lstinline{Biparser}.
%
%
%
\begin{definition}
A biparser \lstinline{p :: Bp u u} is \emph{backward round tripping}
if for all \lstinline{x :: u} and \lstinline{s, s' :: String}
then (recalling that \lstinline{print p :: u -> Maybe (v, String)}):
\vspace{-0.5em}
$$
\lstinline{fmap snd (print p x) = Just s}
\quad \Longrightarrow \quad
\lstinline{parse p (s ++ s') = (x, s')}.
$$
\end{definition}
That is, if a biparser \lstinline{p} when used as a printer (going
backwards) on an input value
\lstinline{x} produces a string \lstinline{s}, then using \lstinline{p} as a parser on a string
with prefix \lstinline{s} and suffix \lstinline{s'}
yields the original input value \lstinline{x} and the remaining input
\lstinline{s'}. 

Note that backward round tripping is
defined for \emph{aligned} biparsers (of type
\lstinline{Bp u u}) since the same value
\lstinline{x} is used as both the input of the printer (typed by the
first type parameter of \lstinline{Bp}) and as the expected output of
the parser (typed by the second type parameter of \lstinline{Bp}).

The dual property is \emph{forward}
round tripping: a source
string \lstinline{s} is parsed (going forwards) into some value \lstinline{x}
which when printed produces the initial source \lstinline{s}:
\begin{definition}
  A biparser \lstinline{p :: Bp u u} is \emph{forward round tripping} if
for every \lstinline{x :: u} and \lstinline{s :: String} we have
that:
\vspace{-0.5em}
$$
\lstinline{parse p s = (x, "")}
\quad \Longrightarrow \quad
\lstinline{fmap snd (print p x) = Just s}
$$
\end{definition}

\vspace{-1em} 
\begin{restatable}{prop}{charRoundtrip}
\label{lem:char-roundtrip}
The biparser \lstinline{char :: Bp Char Char}
(\S\ref{subsec:biparser-promonad}) is both backward and forward
round tripping. Proof by expanding definitions and algebraic reasoning.
\end{restatable}

\newcommand{\propP}[2]{#2 \;\in\; \mathcal{P}_{#1}}
\newcommand{\propQ}[2]{#2 \;\in\; \mathcal{Q}^{#1}}
\newcommand{\propR}[3]{#3 \;\in\; \mathcal{R}^{#1}_{#2}}

Note, in some applications, forward round tripping is too
strong. Here it requires that every printed value corresponds to at most one source string.
This is often not the case as ASTs typically discard
formatting and comments so that pretty-printed
code is lexically different to the original source. However,
different notions of equality enable more reasonable forward round-tripping properties.

Although one can check round-tripping properties of biparsers by expanding
their definitions and the underlying monadic profunctor operations,
a more scalable approach is provided if a round-tripping
property is \emph{compositional} with respect to the monadic
profunctor operations, i.e., if these operations
preserve the property. Compositional properties are easier to enforce and check
since only the individual atomic components need round-tripping proofs.
Such properties are then guaranteed ``by construction'' for programs
built from those components.

\subsection{Compositional properties of monadic bidirectional
  programming}
\label{subsec:compositionality}

Let us first formalize compositionality as follows.
A \emph{property} $\mathcal{R}$ over a monadic profunctor \lstinline{P} is a
family of subsets $\mathcal{R}^{\lstinline{u}}_{\lstinline{v}}$ of
\lstinline{P u v} indexed by types \lstinline{u} and \lstinline{v}.


\begin{definition}
  A property $\mathcal{R}$ over a monadic profunctor \lstinline{P} is
  \emph{compositional} if the monadic profunctor operations are closed over
  $\mathcal{R}$, \ie{}, the following conditions hold for all types \lstinline{u},
  \lstinline{v}, \lstinline{w}:
  \begin{enumerate}[itemsep=-0.4em]
    \item For all \lstinline{x :: v},
      \vspace{-1.25em}
      \begin{equation}\label{eqn:comp-return}\tag{comp-return}
        \propR{\lstinline{u}}{\lstinline{v}}{\lstinline{(return x)}}
      \end{equation}
    \item For all \lstinline{p :: P u v} and \lstinline{k :: v -> P u w},
      \vspace{-0.5em}
      \begin{equation}\label{eqn:comp-bind}\tag{comp-bind}
        \left(
        \propR{\lstinline{u}}{\lstinline{v}}{\lstinline{p}}\right)
        \,\wedge\,
        \left(\forall \lstinline{v} . \; \propR{\lstinline{u}}{\lstinline{w}}{\lstinline{(k v)}}
        \right)
        \; \implies \;
        \lstinline{(p >>= k)} \propR{\lstinline{u}}{\lstinline{w}}{}
      \end{equation}
    \item For all \lstinline{p :: P u' v} and \lstinline{f :: u -> Maybe u'},
    \vspace{-0.5em}
    \begin{equation}\label{eqn:comp-comap}\tag{comp-comap}
      \propR{\lstinline{u'}}{\lstinline{v}}{\lstinline{p}}
      \;\; \implies \;\;
      \propR{\lstinline{u}}{\lstinline{v}}{\lstinline{(comap f p)}}
    \end{equation}
  \end{enumerate}
\end{definition}

\noindent
Unfortunately for biparsers, forward and backward round tripping as defined above
are \emph{not} compositional: \lstinline{return} is not
backward round tripping and \lstinline{>>=} does not preserve forward
round tripping. Furthermore, these two properties are restricted to
biparsers of type \lstinline{Bp u u} (i.e., aligned biparsers)
but compositionality requires that the two type parameters of the monadic
profunctor can differ in the case of \lstinline{comap} and
\lstinline{(>>=)}. This suggests that we need to look for more
general properties that capture the full gamut of possible biparsers.



We first focus on backward round tripping. Informally,
backward round tripping states that if you print (going backwards) and
parse the resulting output (going forwards) then you get back the initial
value. However, in a general biparser \lstinline{p :: Bp u v}, the input
type of the printer \lstinline{u} differs from the output type of the
parser \lstinline{v}, so we cannot compare them. But our intent for printers
is that what we actually print is a fragment of \lstinline{u},
a fragment which is given as the output of the printer. By thus comparing the outputs of both
the parser and printer, we obtain the following variant of backward round
tripping:
\begin{definition}
A biparser \lstinline{p :: Bp u v} is \emph{weak backward round tripping} if
for all \lstinline{x :: u}, \lstinline{y :: v}, and \lstinline{s},
\lstinline{s' :: String}
then:
%
%
\begin{align*}
&  \lstinline{print p x = Just (y, s)} \quad \Longrightarrow \quad
\lstinline{parse p (s ++ s') = (y, s')}
\end{align*}
\end{definition}
Removing backward round tripping's restriction to aligned
biparsers and using the result \lstinline{y :: v} of the printer gives
us a property that \emph{is} compositional:

\begin{restatable}{prop}{biparserwellbehaved}
\label{prop:biparser-well-behaved}
Weak backward round tripping of biparsers is compositional.
\end{restatable}

\begin{proposition}
\label{lemm:char-well-behaved}
The primitive biparser \lstinline{char} is weak backward round tripping.
\end{proposition}


\begin{corollary}\label{prop:string-weak}
Propositions~\ref{prop:biparser-well-behaved} \& \ref{lemm:char-well-behaved} imply
\lstinline{string} is weak backward round tripping.
\end{corollary}
This property is ``weak'' as it does not
constrain the relationship between the input \lstinline{u}
of the printer and its output \lstinline{v}. In fact, there is no hope
for a compositional property to do so: the monadic profunctor
combinators do not enforce a relationship between
them. However, we can regain
compositionality for the stronger backward round-tripping property
by combining the weak compositional property with an additional
non-compositional property on the relationship between the printer's
input and output. This relationship is represented by the function
that results from ignoring the printed string, which amounts to removing the main
effect of the printer. Thus we call this operation a \emph{purification}:
\begin{code}
purify :: forall u v. Bp u v -> u -> Maybe v
purify p u = fmap fst (print p u)
\end{code}
Ultimately, when a biparser is aligned (\lstinline{p :: Bp u u}) we want an input to the
printer to be returned in its output, i.e,
\lstinline{purify p} should equal \lstinline{\x -> Just x}.
If this is the case, we recover the original backward round
tripping property:

\begin{restatable}{thm}{weakToBackward}
\label{thm:weak-to-backward}
  If \lstinline{p :: P u u} is weak backward round tripping, and
  for all \lstinline{x :: u}. \lstinline{purify p x = Just x},
  then \lstinline{p} is backward round tripping.
\end{restatable}
\noindent
Thus, for any biparser \lstinline{p}, we can get backward round
tripping by proving that its atomic subcomponents are weak backward
round tripping, and proving that \lstinline{purify p x = Just x}. The
interesting aspect of the purification condition here is that it
renders irrelevant the domain-specific effects of the biparser,
i.e., those related to manipulating source strings. This
considerably simplifies any proof. Furthermore, the
definition of \lstinline{purify} is a \emph{monadic profunctor
  homomorphism} which provides a set of equations that can be used
to expedite the reasoning.

\begin{definition}
  A \emph{monadic profunctor homomorphism} between monadic profunctors \lstinline{P} and
  \lstinline{Q} is a polymorphic function \lstinline{proj :: P u v -> Q u v}
  such that:
\begin{align*}
  \lstinline{proj (comap}_P \; \lstinline{ f p)} \; & \equiv \;
  \lstinline{comap}_Q \; \lstinline{f (proj p)} \\[-0.2em]
  \lstinline{proj (p >>=}_P \; \lstinline{k)} \; & \equiv \;
  \lstinline{(proj p)} \;\, \lstinline{>>=}_Q \;\, \lstinline{(\\x -> proj (k x))} \\[-0.2em]
  \lstinline{proj (return}_P \; \lstinline{ x)} \; & \equiv \;
    \lstinline{return}_Q \; \lstinline{ x}
\end{align*}
\end{definition}

\begin{proposition}\label{prop:purify-homom}
  The \lstinline{purify :: Bp u v -> u -> Maybe v} operation for
  biparsers (above) is a monadic
  profunctor homomorphism between \lstinline{Bp} and the monadic
  profunctor \lstinline{PartialFun u v = u -> Maybe v}.
\end{proposition}
\begin{corollary}\label{cor:string-backward-rt}
  (of Theorem~\ref{thm:weak-to-backward} with Corollary~\ref{prop:string-weak} and Proposition~\ref{prop:purify-homom})
  The biparser \lstinline{string} is backward round tripping.
\end{corollary}
\begin{proof}
 First prove (in
\iflongVersion Appendix~\ref{app:aux-results}%
\fi%
) the following properties
of biparsers \lstinline{char}, \lstinline{int},
and \lstinline{replicatedBp :: Int -> Bp u v -> Bp [u] [v]}
  (writing \lstinline{proj} for \lstinline{purify}):
  \begin{align}
    \lstinline{proj char n} & \equiv \lstinline{Just n} \label{eq:proj-char} \\[-0.2em]
    \lstinline{proj int n} & \equiv \lstinline{Just n} \label{eq:proj-int} \\[-0.2em]
    \lstinline{proj (replicateBp (length xs) p) xs} & \equiv \lstinline{mapM (proj p) xs}
\label{eq:proj-replicate}
  \end{align}
  From these and the homomorphism properties
  we can prove \lstinline{proj string = Just}:
  \begin{align*}
    & \lstinline{proj string xs}  \\[-0.2em]
    \equiv \;\, &
       \lstinline{proj (comap length int >>= \\n -> replicateBp n
    char) xs} \\[-0.2em]
\textit{Prop.\ref{prop:purify-homom}} \equiv \;\, & \lstinline{(comap length (proj int) >>= \\n -> proj (replicateBp n char)) xs}\\[-0.2em]
(\ref{eq:proj-int}) \, \equiv \;\, & \lstinline{(comap length Just >>= \\n -> proj  (replicateBp n char)) xs} \\[-0.2em]
\textit{Def.\ref{def:promonad}} \, \equiv \;\, & \lstinline{proj (replicateBp (length xs) char) xs} \\[-0.2em]
(\ref{eq:proj-replicate}) \equiv \;\, & \lstinline{mapM (proj char) xs} \\[-0.2em]
    (\ref{eq:proj-char}) \equiv \;\, & \lstinline{mapM Just xs} \\[-0.2em]
    \{\textit{monad}\} \equiv \;\, & \lstinline{Just xs}
  \end{align*}
  Combining \lstinline{proj string = Just} with
  Corollary~\ref{prop:string-weak} (\lstinline{string} is weak backward round
  tripping) enables Theorem~\ref{thm:weak-to-backward}, proving that \lstinline{string} is
  backward round tripping.
\end{proof}

The other two core examples in this paper also permit a definition of
\lstinline{purify}.  We capture the general pattern as follows:
\begin{definition}
  \label{def:purifiable}
  A \emph{purifiable monadic profunctor} is a monadic
  profunctor \lstinline{P} with a homomorphism \lstinline{proj} from \lstinline{P} to the
  monadic profunctor of partial functions \lstinline{- -> Maybe -}.
  We say that \lstinline{proj p} is the \emph{pure projection} of \lstinline{p}.
\end{definition}
\begin{definition}
A pure projection \lstinline{proj p :: u -> Maybe v} is
called the \emph{identity projection} when \lstinline{proj p x = Just x}
 for all \lstinline{x :: u}.
\end{definition}

\noindent
Here and in Sections~\ref{sec:lenses} and \ref{sec:generators}, identity projections
enable compositional round-tripping properties to
be derived from more general non-compositional
properties, as seen above for backward round tripping of biparsers.


We have neglected forward round tripping, which is not
compositional, not even in a weakened form.
However, we can generalise compositionality
with conditions related to \emph{injectivity}, enabling a
generalisation of forward round tripping. We call the generalised
meta-property \emph{quasicompositionality}.

\subsection{Quasicompositionality for monadic profunctors}
\label{subsec:quasicompositionality}

\noindent
An injective function $f : A \rightarrow B$ is a function for which
there exists a left inverse $f^{-1} : B \rightarrow A$, i.e., where
$f^{-1} \circ f = id$. We can see this pair of functions as a simple
kind of bidirectional program, with a forward round-tripping property
(assuming $f$ is the forwards direction). We can lift the notion
of injectivity to the monadic profunctor setting and capture forward round-tripping
properties that are preserved by the monadic profunctor operations,
given some additional injectivity-like restriction. We first formalise
the notion of an \emph{injective arrow}.

Informally, an injective arrow \lstinline{k :: v -> m w} produces an output
from which the input can be recalculated:


\begin{definition}
 \label{def:injective-arrow}
Let \lstinline{m} be a monad. A function \lstinline{k :: v -> m w}
is an \emph{injective arrow} if there exists
\lstinline{k' :: w -> v} (the \emph{left arrow inverse} of
\lstinline{k}) such that for all \lstinline{x :: v}:
\begin{align*}
\lstinline{k x >>= \\y -> return (x, y)} \;\; \equiv \;\;
\lstinline{k x >>= \\y -> return (k' y, y)}
\end{align*}
\end{definition}

Next, we define \emph{quasicompositionality} which extends the
compositionality meta-property with the requirement for \lstinline{>>=} to be
applied to injective arrows:

\begin{definition}
  \label{def:quasicomp}
  Let \lstinline{P} be a monadic profunctor.
  A property $\mathcal{R}^{\lstinline{u}}_{\lstinline{v}} \subseteq \lstinline{P u v}$
  indexed by types \lstinline{u} and \lstinline{v}
  is \emph{quasicompositional}
  if the following holds
  \begin{enumerate}[itemsep=-0.4em]
    \item
      For all \lstinline{x :: v},%
       \vspace{-1.25em}
      \begin{equation}\label{eqn:qcomp-return}\tag{qcomp-return}
        \propR{\lstinline{u}}{\lstinline{v}}{\lstinline{(return x)}}
      \end{equation}
    \item
      For all \lstinline{p :: P u v},\;
      \lstinline{k :: v -> P u w},\;
      {\bf if \lstinline{k} is an injective arrow,}
      \vspace{-0.6em}
      \begin{equation}\label{eqn:qcomp-bind}\tag{qcomp-bind}
        \left(\propR{\lstinline{u}}{\lstinline{v}}{\lstinline{p}}\right) \,\wedge\,
        \left(\forall \lstinline{v} . \; \propR{\lstinline{u}}{\lstinline{w}}{(\lstinline{k v})}\right)
        \; \implies \;
        \propR{\lstinline{u}}{\lstinline{w}}{\lstinline{(p >>= k)}}
      \end{equation}
    \item
      For all \lstinline{p :: P u' v},\;
      \lstinline{f :: u -> Maybe u'},
      \vspace{-0.6em}
      \begin{equation}\label{eqn:qcomp-comap}\tag{qcomp-comap}
        \propR{\lstinline{u'}}{\lstinline{v}}{\lstinline{p}} \wedge
        \Longrightarrow
        \propR{\lstinline{u}}{\lstinline{w}}{\lstinline{(comap f p)}}
      \end{equation}
  \end{enumerate}
\end{definition}


\noindent
We now formulate a weakening of forward round tripping. 
As with weak backward round tripping, we rely on the
idea that the printer \emph{outputs} both a string and the value that was
printed, so that we need to compare the outputs of both the parser and the
printer, as opposed to comparing the output of the parser with the input of the
printer as in (strong) forward round tripping. If running the parser component of
a biparser on a string \lstinline{s01} yields a value \lstinline{y} and a
remaining string \lstinline{s1}, and the printer outputs that same
value \lstinline{y} along with a string \lstinline{s0}, then \lstinline{s0}
is the prefix of \lstinline{s01} that was consumed by the parser,
i.e., \lstinline{s01 = s0 ++ s1}.
\begin{definition}
A biparser \lstinline{p : Bp u v} is \textit{weak forward round
  tripping} if for all \lstinline{x :: u}, \lstinline{y :: v},
and \lstinline{s0, s1, s01 :: String} then:
\begin{equation*}
  \lstinline{parse p s01 = (y, s1)}
\;\,\wedge\;\,
  \lstinline{print p x = Just (y, s0)}
\;\; \Longrightarrow \;\;
  \lstinline{s01 = s0 ++ s1}
\end{equation*}
\end{definition}
%

\vspace{-1em} 
\begin{restatable}{prop}{wfquasi}
  \label{prop:wf-quasi}
  Weak forward round tripping is quasicompositional.
\end{restatable}

\begin{proof}
  We sketch the \ref{eqn:qcomp-bind} case, where \lstinline{p = (m >>= k)}
  for some \lstinline{m} and \lstinline{k} that are weak forward roundtripping.
  From \lstinline{parse (m >>= k) s01 = (y, s1)}, it follows that there exists
  \lstinline{z}, \lstinline{s} such that \lstinline{parse m s01 = (z, s)} and
  \lstinline{parse (k z) s = (y, s1)}. Similarly \lstinline{print (m >>= k) x = Just (y, s0)}
  implies there exists \lstinline {z'}, \lstinline{s0'} such that
  \lstinline{print m x} \lstinline{=} \lstinline{Just (z', s0')} and
  \lstinline{print (k z') x} \lstinline{=} \lstinline{Just (y, s1')}
  and \lstinline{s0} \lstinline{=} \lstinline{s0' ++ s1'}. Because \lstinline{k} is an injective arrow,
  we have \lstinline{z = z'} (see appendix).
  We then use the assumption that
  \lstinline{m} and \lstinline{k} are weak forward roundtripping on \lstinline{m}
  and on \lstinline{k a},
  and deduce that \lstinline{s01 = s0' ++ s} and \lstinline{s = s1' ++ s1}
  therefore \lstinline{s01 = s0 ++ s1}.
\end{proof}

\begin{proposition}\label{prop:wf-quasi-char}
The \lstinline{char} biparser is weak forward
round tripping.
\end{proposition}

\begin{corollary}\label{prop:wf-quasi-string}
Propositions~\ref{prop:wf-quasi} and \ref{prop:wf-quasi-char}
imply that \lstinline{string} is weak forward round tripping
if we restrict the parser to inputs whose digits
do not contain redundant leading zeros.
\end{corollary}

\begin{proof}
  All of the right operands of \lstinline{>>=} in the definition of
  \lstinline{string} are injective arrows, apart from
\lstinline{\ds -> return (read ds)}
at the end of the auxiliary \lstinline{int}
  biparser. Indeed, the \lstinline{read} function is not injective
  since multiple strings may parse to the same integer:
  \lstinline{read "0" = read "00" = 0}. But the pre-condition to
  the proposition (no redundant leading zero digits) restricts the input
  strings so that \lstinline{read} is injective.  The rest of the
  proof is a corollary of
  Propositions~\ref{prop:wf-quasi} and \ref{prop:wf-quasi-char}.
\end{proof}

Thus, quasicompositionality gives us scalable reasoning for
weak forward round tripping, which is by construction for biparsers: we
just need to prove this property for individual atomic biparsers.
Similarly to backward round tripping, we can prove forward
round tripping by combining weak forward round tripping with
the identity projection property:
\begin{theorem}
If \lstinline{p :: P u u} is weak forward round-tripping, and for all
\lstinline{x :: u}, \lstinline{purify p x = Just x}, then
\lstinline{p} is forward round tripping.
\end{theorem}
\begin{corollary}
The biparser \lstinline{string} is forward round tripping
by the above theorem (with identity projection
shown in the proof of Corollary~\ref{cor:string-backward-rt}) and
Corollary~\ref{prop:wf-quasi-string}.
\end{corollary}


\noindent
In summary, for any BX we can consider two round-tripping properties: fowards-then-backwards and
backwards-then-forwards, called just \emph{forward} and
\emph{backward} here respectively.
Whilst combinator-based approaches can guarantee round-tripping
by construction, we have made a trade-off to get greater expressivity
in the monadic approach. However, we regain the
ability to reason about bidirectional transformations in a manageable, scalable
way if round-tripping properties are compositional. Unfortunately, due to the monadic
profunctor structuring, this tends not to be the case. Instead,
weakened round-tripping properties can be compositional or quasicompositional (adding
injectivity). In such cases, we recover the stronger
property by proving a simple property on aligned transformations: that
the backwards direction faithfully reproduces its input as its output
(\emph{identity projection}).
\iflongVersion
Appendix~\ref{app:manual-parser-printer}
\fi
compares this reasoning approach to a proof of backwards round
tripping for separately implemented parsers and printers (not
using our combined monadic approach).


\section{Monadic bidirectional programming for lenses}
\label{sec:lenses}
Lenses are a common object of study in bidirectional programming,
comprising a pair of functions
$(\lstinline{get} : S \rightarrow V, \lstinline{put} : V \rightarrow S
\rightarrow S)$ satisfying \emph{well-behaved lens} laws
shown in Section~\ref{sec:introduction}. Previously, when considering
the monadic structure of parsers and printers, the starting point
was that parsers already have a well-known monadic structure. The
challenge came in finding a reasonable monadic characterisation for
printers that was compatible with the
parser monad. In the end, this construction was
expressed by a product of two monadic profunctors \lstinline{Fwd m} and
\lstinline{Bwd n} for monads \lstinline{m} and  \lstinline{n}. For
lenses we are in the same position: the forwards
direction (\lstinline{get}) is already a monad---the reader
monad. The backwards direction \lstinline{put} is not a monad
since it is contravariant in its parameter; the same situation as
printers. We can apply the same approach of
``monadisation'' used for parsers and printers, giving the following new
data type for lenses:
\begin{edoc}
data L s u v = L { get :: s -> v, put :: u -> s -> (v, s) }
\end{edoc}
The result of \lstinline{put} is paired with a
covariant parameter \lstinline{v} (the result type of \lstinline{get})
in the same way as monadic
printers. Instead of mapping a view and a source to a
source, \lstinline{put} now maps values of a
different type \lstinline{u}, which we call a \emph{pre-view},
along with a source \lstinline{s} into a pair of a view \lstinline{v} and source
\lstinline{s}. This definition can
be structured as a monadic profunctor via a pair of \lstinline{Fwd} and \lstinline{Bwd}
constructions:
\begin{edoc}
type L s = (Fwd (Reader s)) :*: (Bwd (State s))
\end{edoc}
Thus by the results of Section~\ref{sec:monadic-profunctors}, we now
have a monadic profunctor characterisation of lenses that allows
us to compose lenses via the monadic interface.

Ideally, \lstinline{get} and \lstinline{put} should be total, but this
is impossible without a way to restrict the domains. In particular,
there is the known problem of ``duplication''~\citep{mu2004}, where
source data may appear more than once in the view, and a necessary
condition for \texttt{put} to be well-behaved is that the duplicates
remain equal amid view updates. This problem is inherent to all
bidirectional transformations, and bidirectional languages have to
rule out inconsistent updates of duplicates either
statically~\citep{FGMPS07} or dynamically~\citep{mu2004}. To remedy
this, we capture both partiality of \lstinline{get} and a predicate on
sources in \lstinline{put} for additional dynamic checking.  This is
provided by the following \lstinline{Fwd} and \lstinline{Bwd} monadic
profunctors:
\begin{edoc}
type ReaderT r m a = r -> m a
type StateT s m a  = s -> m (a, s)
type WriterT w m a = m (a, w)

type L s = (Fwd (ReaderT s Maybe))
       :*: (Bwd (StateT s (WriterT (s -> Bool) Maybe)))

-- Smart deconstructors:
get :: L s u v -> (s -> Maybe v)
put :: L s u v -> (u -> s -> Maybe ((v, s), s -> Bool))
\end{edoc}
Going forwards, \emph{getting} a view \lstinline{v}
from a source \lstinline{s} may fail if there is no view for the
current source.
Going backwards, \textit{putting} a pre-view \lstinline{u} updates some source
\lstinline{s} (via the state transformer \lstinline{StateT s}), but
with some further structure returned, provided by \lstinline{WriterT (s -> Bool) Maybe}
 (similar to the writer transformer used for biparsers,
 \S~\ref{subsec:biparser-promonad}, p.~\pageref{subsec:biparser-promonad}).
The \lstinline{Maybe} here captures the possibility that
\lstinline{put} can fail. The \lstinline{WriterT (s -> Bool)} structure provides
a predicate which detects the ``duplication''
issue mentioned earlier. Informally, the predicate can be used to check that previously
modified locations in the source are not modified again. For example,
if a lens has a source made up of a bit vector, and a \lstinline{put}
sets bit $i$ to 1, then the returned predicate will return \lstinline{True} for all
bit vectors where bit $i$ is 1, and \lstinline{False} otherwise. This predicate can
then be used to test whether further \lstinline{put} operations on the
source have modified bit $i$.

Similarly to biparsers, a pre-view \lstinline{u} can be understood
as \emph{containing} the view \lstinline{v} that is to be merged with the
source, and which is returned with the updated source.
Ultimately, we wish to form lenses of matching input and output types
(i.e. \lstinline{L s v v}) satisfying the standard lens
well-behavedness laws, modulo explicit management
of partiality via \lstinline{Maybe} and testing for conflicts via the
predicate:
%
\begin{align*}
  \hspace{-1.5em}
\tag{L-PutGet}\label{eqn:L-PutGet}
  \lstinline{put l x s = Just ((_, s'), p')}
   \; \wedge \; \lstinline{p' s'}
 \implies
\lstinline{get l s' = Just x} \\
  \tag{L-GetPut}\label{eqn:L-GetPut}
  \lstinline{get l s = Just x} \quad\implies\quad
  \lstinline{put l x s = Just ((_, s), _)}
\end{align*}
\ref{eqn:L-PutGet} and \ref{eqn:L-GetPut} are
backward and forward round tripping respectively.
Some lenses, such as the later example, are not defined for all views.
In that case we may say that the lens is backward/forward round tripping
in some subset $\lstinline{P} \subseteq \lstinline{u}$ when the above properties
only hold when \lstinline{x} is an element of \lstinline{P}.

For every source type \lstinline{s},
the lens type \lstinline{L s} is automatically a monadic profunctor by its
definition as the pairing of \lstinline{Fwd} and \lstinline{Bwd}
(Section~\ref{sec:constructing-mp}), and the following instance of \lstinline{MonadPartial} for
handling failure and instance of \lstinline{Monoid} to satisfy the requirements of the writer monad:
\begin{code}
instance MonadPartial (StateT s (WriterT (s -> Bool) Maybe)) where
  toFailure Nothing  = StateT (\_ -> WriterT Nothing)
  toFailure (Just x) = StateT (\s -> WriterT (Just ((x , s), mempty)))
\end{code}
\begin{edoc}
instance Monoid (s -> Bool) where
  mempty      = \_ -> True
  mappend h j = \s0 -> h s0 && j s0
\end{edoc}

\noindent
A simple lens example operates on key-value
maps. For keys of type \lstinline{Key} and values of type
\lstinline{Value}, we have the following source type and a simple lens:
\begin{code}
type Src = Map Key Value
atKey :: Key -> L Src Value Value  -- Key-focussed lens
atKey k = mkLens (lookup k)
  (\v -> \map -> Just ((v, insert k v map), \m' -> lookup k m' == Just v))
\end{code}
%
The \lstinline{get} component of the \lstinline{atKey} lens does a
lookup of the key \lstinline{k} in a map, producing \lstinline{Maybe}
of a \lstinline{Value}. The \lstinline{put} component inserts a value
for key \lstinline{k}. When the key already exists, \lstinline{put}
overwrites its associated value.

Due to our approach,
multiple calls to \lstinline{atKey} can be composed monadically,
giving a lens that gets/sets multiple key-value pairs at once.
The list of keys and the list of values are passed separately,
and are expected to be the same length.
\begin{edoc}
atKeys :: [Key] -> L Src [Value] [Value]
atKeys [] = return []
atKeys (k : ks) = do
  x  <- comap headM (atKey k)     -- headM :: [a] -> Maybe a
  xs <- comap tailM (atKeys ks)   -- tailM :: [a] -> Maybe [a]
  return (x : xs)
\end{edoc}
\noindent We refer interested readers
to our implementation~\citep{lib} for more examples, including further examples involving trees.

\paragraph{Round tripping}

We apply the reasoning framework of
Section~\ref{sec:compositionality}, taking
the standard lens laws as the starting point
(neither of which are compositional).


We first weaken backward round tripping to be compositional.
Informally, the property expresses the idea, that
if we \lstinline{put} some value \lstinline{x} in a source \lstinline{s},
resulting in a source
\lstinline{s'}, then what we \lstinline{get} from \lstinline{s'} is \lstinline{x}.
However two important changes are needed to adapt to our generalised type
of lenses and to ensure compositionality.
First, the value \lstinline{x} that was put is now to be found in the output of
\lstinline{put}, whereas there is no way to constrain the input of \lstinline{put}
because its type \lstinline{v} is abstract.
Second, by sequentially composing lenses such as in \lstinline{l >>= k}, the
output source \lstinline{s'} of \lstinline{put l} will be further modified by
\lstinline{put (k x)}, so this round-tripping property must constrain all
potential modifications of \lstinline{s'}. In fact, the predicate \lstinline{p}
ensures exactly that the view \lstinline{get l} has not changed and is
still \lstinline{x}. It is not even necessary to refer to \lstinline{s'},
which is just one source for which we expect \lstinline{p} to be
\lstinline{True}.
\begin{definition}
A lens \lstinline{l :: L s u v} is \emph{weak backward round tripping} if
for all \lstinline{x :: u}, \lstinline{y :: v}, for all sources \lstinline{s},
  \lstinline{s'}, and for all \lstinline{p :: s -> Bool}, we have:
\begin{align*}
  & \lstinline{put l x s = Just ((y, _), p)} \;\wedge\;
\lstinline{p s'}
\quad \Longrightarrow \quad
\lstinline{get l s' = Just y}
\end{align*}
\end{definition}
\begin{theorem}
  Weak backward round tripping is a compositional property.
\end{theorem}

Again, we complement this weakened version of round tripping with the notion of
purification.
\begin{proposition}
  Our lens type \lstinline{L} is a \emph{purifiable} monadic profunctor
  (Definition~\ref{def:purifiable}), with a family of pure projections
  \lstinline{proj s} indexed by a source \lstinline{s}, defined:
\begin{code}
proj :: s -> L s u v -> (u -> Maybe v)
proj s = \l u -> fmap (fst . fst) (put l u s)
\end{code}
\end{proposition}

\vspace{-1.5em} 

\begin{restatable}{thm}{thmWbRt}
\label{thm:wb-to-b-rt}
  If a lens \lstinline{l :: L s u u} is weak backward round tripping and has
  identity projections on some subset $\lstinline{P} \subseteq
  \lstinline{u}$ (i.e., for all
  \lstinline{s}, \lstinline{x} then $\lstinline{x} \in \lstinline{P}
  \Rightarrow \lstinline{proj s l x = Just x}$) then
  \lstinline{l} is also backward round tripping on all $\lstinline{x} \in \lstinline{P}$.
\end{restatable}
\noindent
To demonstrate,
we apply this result to
\lstinline{atKeys :: [Key] -> L Src [Value] [Value]}.

\begin{restatable}{prop}{atkeyWb}
\label{prop:atkey-wb}
  The lens \lstinline{atKey k} is weak backward round tripping.
\end{restatable}

\vspace{-1.25em} 

\begin{restatable}{prop}{atkeyPp}
\label{prop:atkey-pp}
  The lens \lstinline{atKey k} has identity projection:
\lstinline{proj z} \lstinline{(atKey k)}\lstinline{=}\lstinline{Just}.
\end{restatable}
Our lens \lstinline{atKeys ks} is therefore weak backward round tripping by construction.
We now interpret/purify \lstinline{atKeys ks} as a partial function, which is
actually the identity function when restricted to lists of the same length as
\lstinline{ks}.
\begin{restatable}{prop}{atkeysPp}
  \label{prop:atkeys-pp}
  For all \lstinline{vs :: [Value]} such that \lstinline{length vs = length ks},
  and for all \lstinline{s :: Src} then
  \lstinline{proj s (atKeys ks) vs = Just vs}.
\end{restatable}
%
%
\begin{corollary}
  By the above results,
\lstinline{atKeys ks :: L Src [Value] [Value]} for all \lstinline{ks}
is backward round tripping on lists of length \lstinline{length ks}.
\end{corollary}

The other direction, forward round tripping, follows a similar story. We first
restate it as a quasicompositional property.
\begin{definition}
A lens \lstinline{l :: L s u v} is \emph{weak forward round tripping} if
for all \lstinline{x :: u}, \lstinline{y :: v}, for all sources \lstinline{s},
\lstinline{s'}, and for all \lstinline{p :: s -> Bool}, we have:
\begin{align*}
  &\lstinline{get l s = Just y}  \;\wedge\;
\lstinline{put l x s = Just ((y, s'), _)}
\quad \Longrightarrow \quad
\lstinline{s = s'}
\end{align*}
\end{definition}
\begin{theorem}
  Weak forward round tripping is a quasicompositional property.
\end{theorem}
Along with identity projection, this gives the original forward
\ref{eqn:L-GetPut} property.

\vspace{-0.25em}
\begin{restatable}{thm}{thmWfRt}
\label{thm:wf-to-f-rt}
  If a lens \lstinline{l} is weak forward round tripping and has
  identity projections on some subset \lstinline{P} (i.e., for all
  \lstinline{s}, \lstinline{x} then
  $\lstinline{x} \in \lstinline{P} \Rightarrow \lstinline{proj s l x = Just x}$) then
  \lstinline{l} is also forward round tripping on \lstinline{P}.
\end{restatable}
We can thus apply this result to our example (details omitted).
\begin{proposition}
For all \lstinline{ks}, the lens \lstinline{atKeys ks :: L Src [Value] [Value]}
is forward round tripping on lists of length \lstinline{length ks}.
\end{proposition}
\vspace{-1em}

\section{Monadic bidirectional programming for generators}
\label{sec:generators}

\noindent
Lastly, we capture the novel notion of \emph{bidirectional
  generators} (\emph{bigenerators}) extending random generators in
property-based testing frameworks like
\emph{QuickCheck} \citep{Claessen2000} to a bidirectional setting. The
forwards direction generates values conforming to a specification;
the backwards direction checks whether values conform to a
predicate. We capture the two together via our monadic profunctor pair as:
\begin{edoc}
type G = (Fwd Gen) :*: (Bwd Maybe)
-- ... with deconstructors and constructors
generate :: G u v -> Gen v                      -- forward direction
check    :: G u v -> u -> Maybe v               -- backward direction
mkG      :: Gen v -> (u -> Maybe v) -> G u v
\end{edoc}
%
The forwards direction of a bigenerator is a generator, while the backwards
direction is a partial function \lstinline{u -> Maybe v}.
A value \lstinline{G u v} represents a subset of \lstinline{v},
where \lstinline{generate} is a generator of values in that subset and
\lstinline{check} maps pre-views \texttt{u} to members of the
generated subset. In the backwards direction,
\lstinline{check g} defines a predicate on \lstinline{u}, which is true
if and only if \lstinline{check g u} is \lstinline{Just} of some
value. The function \lstinline{toPredicate} extracts this predicate
from the backward direction:
\begin{code}
toPredicate :: G u v -> u -> Bool
toPredicate g x = case check g x of Just _  -> True; Nothing -> False
\end{code}
The bigenerator type \lstinline{G} is automatically a
monadic profunctor due to our construction (\S\ref{sec:monadic-profunctors}). Thus, monad and profunctor
instances come for free, modulo (un)wrapping of
 constructors and given a trivial instance of
\lstinline{MonadPartial}:
\begin{code}
instance MonadPartial Maybe where toFailure = id
\end{code}
Due to space limitations, we refer readers to
\iflongVersion Appendix~\ref{sec:generator_example}
\fi
for an example of a
compositionally-defined bigenerator that produces binary search trees.

\paragraph{Round tripping}


A random generator can be interpreted as the set of values it
may generate, while a predicate represents the set of values satisfying it.
For a bigenerator \lstinline{g},
we write $\lstinline{x}\in \lstinline{generate g}$ when
\lstinline{x} is a possible output of the generator.
The generator of a bigenerator \lstinline{g}
should match its predicate \lstinline{toPredicate g}. This requirement
equates to round-tripping properties: a
bigenerator is \emph{sound} if every value which it can generate satisfies the
predicate (forward round tripping); a bigenerator is \emph{complete} if every value which satisfies the
predicate can be generated (backward round tripping). Completeness is often
more important than soundness in testing because unsound tests can be
filtered out by the predicate, but completeness determines
the potential adequacy of testing.
\begin{definition}
  A bigenerator \lstinline{g :: G u u} is \emph{complete} (backward
  round tripping) when
\lstinline{toPredicate g x = True} implies
$\lstinline{x} \in \lstinline{generate g}$.
\end{definition}
\begin{definition}
A bigenerator \lstinline{g :: G u u} is \emph{sound}  (forward
  round tripping) if
for all \lstinline{x :: u},
$\lstinline{x} \in \lstinline{generate g}$
implies that
\lstinline{toPredicate g x = True}.
\end{definition}

\noindent
Similarly to backward round tripping of biparsers and lenses,
completeness can be split into a compositional weak completeness and a
purifiable property.

As before, the compositional weakening of completeness relates the forward and
backward components by their outputs, which have the same type.
\begin{definition}
  A bigenerator \lstinline{g :: G u v} is \emph{weak-complete} when
  \[
    \lstinline{check g x = Just y}
    \;\Longrightarrow\;
    {\lstinline{y}\in\lstinline{generate g}}.
  \]
\end{definition}
\begin{theorem}
  Weak completeness is compositional.
\end{theorem}
In a separate step, we connect the input of the backward direction,
\textit{i.e.}, the checker, by reasoning directly about its pure
projection (via a more general form of identity projection)
which is defined to be the checker itself:
\begin{theorem}
  A bigenerator \lstinline{g :: G u u} is complete if it is weak-complete and its
  checker satisfies a pure projection property:
  $\lstinline{check g x = Just x'} \, \Rightarrow \, \lstinline{x = x'}$
\end{theorem}
Thus to prove completeness of a bigenerator \lstinline{g :: G u u},
we first have weak-completeness by construction, and we can then show that
\lstinline{check g} is a restriction of the identity function, interpreting
all bigenerators simply as partial functions.

Considering the other direction, soundness, there is unfortunately no
decomposition into a quasicompositional property and a property on
pure projections. To see why, let \lstinline{bool} be a random uniform
bigenerator of booleans, then consider for example,
\lstinline{comap isTrue bool} and \lstinline{comap isTrue (return True)},
where \lstinline{isTrue True = Just True} and \lstinline{isTrue False = Nothing}.
Both satisfy any quasicompositional property satisfied by \lstinline{bool},
and both have the same pure projection \lstinline{isTrue}, and yet the former
is unsound---it can generate \lstinline{False}, which is rejected by
\lstinline{isTrue}---while the latter is sound. This is not a problem in practice, as unsoundness, especially in small scale, is inconsequential in testing.
But it does raise an intellectual challenge and an interesting point in the design space, where
ease of reasoning has been traded for greater expressivity in the monadic approach.

\section{Discussion and Related Work}
\label{sec:discussion}
Bidirectional transformations are a widely applicable technique
 used in many domains~\citep{Czarnecki2009}. Among language-based
solutions, the lens framework is most
influential~\citep{FGMPS07,BaCFGP10,BoFPPS08,FoMV12,RaLFC13,PaHF14}.
 Broadly speaking, 
 combinators are used as programming constructs
 with which complex lenses are created by combining simpler ones.
 The combinators preserve round tripping,
 and therefore the resulting programs are correct by construction.
 A problem with lens languages is that they tend to be disconnected from
 more general programming. Lenses can only be constructed by very
 specialised combinators and are not subject to existing abstraction
 mechanisms.
 Our approach allows bidirectional
 transformations to be built using standard components of functional
 programming, and gives a reasoning framework for studying
 compositionality of round-tripping properties.

The framework of \emph{applicative lenses}~\citep{MaWa15}
uses a function representation of lenses to
lift the point-free restriction of the combinator-based languages, and enables bidirectional programming with
explicit recursion and pattern matching. Note that the use of
``applicative'' in applicative lenses refers to the transitional sense
of programming with $\lambda$-abstractions and functional applications,
which is not directly related to applicative functors. In a subsequent work, the authors developed a language known
as HOBiT~\citep{esop18}, which went further in featuring
proper binding of variables. Despite the
success in supporting $\lambda$-abstractions and function applications in programming bidirectional transformations, none of the languages
have explored advanced patterns such as monadic programming.

The work on \emph{monadic lenses}~\citep{Abou-Saleh2016} investigates lenses
with effects. For instance, a ``put'' could require additional input to resolve
conflicts. 
Representing effects with monads helps reformulate the laws of round-tripping.
In contrast, we made the type of lenses itself a monad, and showed how
they can be composed monadically. 
Our method is applicable to monadic lenses, yielding what one might call
\emph{monadic monadic lenses}: monadically composable lenses with monadic effects.
We conjecture that laws for monadic lenses can be adapted to this setting
with similar compositionality properties, reusing our reasoning framework.

Other work leverages profunctors for bidirectionality.
Notably, a
\emph{Profunctor optic}~\citep{poptics} between a source type \lstinline{s}
and a view type \lstinline{v} is a function of type \lstinline{p v v -> p s s},
for an abstract profunctor \lstinline{p}.
Profunctor optics and our monadic profunctors offer orthogonal composition patterns:
profunctor optics can be composed ``vertically'' using function composition,
whereas monadic profunctor composition is ``horizontal'' providing
sequential composition. In both cases, composition
in the other direction can only be obtained by breaking the abstraction.


It is folklore in the Haskell community that profunctors can be combined with applicative functors~\citep{McPa08}. The pattern
is sometimes called a \emph{monoidal} profunctor.
The \textsf{codec} library~\citep{codec} mentioned in
Section~\ref{sec:monadic-profunctors} prominently features two applications of
this applicative programming style: binary serialisation (a form of
parsing/printing) and conversion to and from JSON structures (analogous to
lenses above). Opaleye~\citep{opaleye}, an EDSL of SQL queries for Postgres databases,
uses an interface of monoidal profunctors to implement generic operations such
as transformations between Haskell datatypes and database queries and responses.

Our framework adapts gracefully to applicative programming,
a restricted form of monadic programming.
By separating the input type from the output type, we can reuse the existing
interface of applicative functors without modification.
Besides our generalisation to monads, purification
and verifying round-tripping properties via (quasi)compositionality are novel in our framework.

Rendel and Ostermann proposed an interface for programming parsers and printers together~\citep{ReOs10}, but they were unable to reuse the existing structure of \lstinline{Functor}, \lstinline{Applicative} and \lstinline{Alternative} classes (because of the need to handle types that are both covariant and contravariant), and had to reproduce the
entire hierarchy separately. In contrast, our approach reuses the
standard type class hierarchy, further extending the
expressive power of bidirectional programming in Haskell. FliPpr~\citep{wang2013flippr, Matsuda2018} is an invertible language that
generates a parser from a definition of a pretty printer. In this
paper, our biparser definitions are more similar to those of parsers
than printers. This makes sense as it has been established that many
parsers are monadic. Similar to the case of HOBiT, there is no
discussion of monadic programming in the FliPpr work.


Previous approaches to unifying random generators and predicates
mostly focused on deriving generators from predicates.
One general technique evaluates predicates lazily to
drive generation (random or enumerative)~\citep{korat,Claessen2015},
but one loses control over the resulting distribution of generated values.
\emph{Luck}~\citep{Lamp17} is a domain-specific language blending narrowing and
constraint solving to specify generators as predicates with
user-provided annotations to control the probability distribution.
In contrast, our programs can be viewed as generators annotated
with left inverses with which to derive predicates. This reversed perspective
comes with trade-offs: high-level properties would be more naturally
expressed in a declarative language of predicates, whereas it is \emph{a priori} more
convenient to implement complex generation strategies in a specialised
framework for random generators.

\paragraph{Conclusions}
\label{sec:conclusions}
This paper advances the expressive power of bidirectional programming; we
showed that the classic bidirectional patterns of parsers/printers and lenses
can be restructured in terms of \emph{monadic profunctors} to
provide sequential composition, with associated reasoning techniques.
This opens up a new area in the design of
embedded domain-specific languages for BX programming, that does not restrict programmers to stylised interfaces.
Our example of bigenerators broadened the scope of BX programming from
transformations (converting between two data representations) to
non-transformational applications.

To demonstrate the applicability of our approach to real code,
we have developed two bidirectional libraries~\cite{lib},
one extending the \textsf{attoparsec} monadic parser combinator library to
biparsers and one extending \textsf{QuickCheck} to bigenerators.
One area for further work is studying
biparsers with \emph{lookahead}. Currently lookahead
can be expressed in our extended \textsf{attoparsec}, but understanding
its interaction with (quasi)compositional
round-tripping is further work.


However, this is not the final word on sequentially composable BX
programs. In all three applications, round-tripping properties are similarly
split into weak round tripping, which is weaker than the original property but
compositional, and purifiable, which is equationally friendly.
An open question is whether an underlying structure can be formalised,
perhaps based on an adjunction model, that captures
bidirectionality even more concretely than monadic profunctors.

\paragraph{Acknowledgments}

We thank the anonymous reviewers for their helpful comments.
The second author was supported partly by EPSRC grant EP/M026124/1.


\bibliography{references}

\iflongVersion
\appendix
\section{Further code}
\label{app:further-code}

\paragraph{Complete \lstinline{Monad} instance for biparsers}
\label{app:mbx-code}

The instance is the straightforward product of the monad instances
for \lstinline{Parser} and \lstinline{Printer}, where the two parts
remain independent:
\begin{code}
instance Monad (Biparser u) where
  return :: v -> Biparser u v
  return v = Biparser (\s -> (v, s)) (\ _ -> (v, ""))

  (>>=) :: Biparser u v -> (v -> Biparser u w) -> Biparser u w
  pu >>= kw = Biparser parse' print' where
    parse'  s = let (v, s') = parse pu s in parse (kw v) s'
    print'  u = let (v, s) = print pu u
                    (w, s') = print (kw v) u in (w, s ++ s')
\end{code}

\section{Proofs for compositional reasoning}
\label{app:aux-results}

\noindent
The supplementary Coq proofs formalise many results of
Section~\ref{sec:compositionality}. We include some results here
as hand-proofs for human consumption.

\charRoundtrip*
\begin{proof}
Expanding the definitions, for backward round tripping: \\
  \lstinline{fmap snd (print p x) = Just [c]},
then \lstinline{parse p ([c] + +s') = (c, s')} (QED).

For forward round tripping, \lstinline{parse p s = (x, "")} means that
s must be \lstinline{[x]}, then:
\lstinline{fmap snd (print p x) = Just [x]} (QED).
\end{proof}

\begin{proposition}
  \label{prop:return-backwards}
  The \lstinline{return} operation for the \lstinline{Biparser}
  monadic profunctor is not backward round tripping, but it is
  \emph{weak} backward round tripping.
\end{proposition}

\begin{proof}
  Let \lstinline{x, y :: u} and \lstinline{s, s' :: String}:
  \begin{itemize}
    \item (Proof that it is \emph{not} backwards round tripping)
    \begin{align*}
             & \lstinline{fmap snd (print (return y) x)} \\
      \equiv\;\; & \lstinline{fmap snd ((\\_ -> Just (y, "")) x)} \\
      \equiv\;\; & \lstinline{fmap snd (Just (y, ""))} \\
      \equiv\;\; & \lstinline{Just ""}
    \end{align*}
    Thus \lstinline{s == ""}. Now we must prove the consequent of
    backwards round tripping, but it turns out to be false:
    \begin{align*}
            & \lstinline{parse (return y) ("" ++ s')} \\
      \equiv\;\; & \lstinline{(\\s -> (y, s)) s'} \\
      \equiv\;\; & \lstinline{(y, s')} \\
      \not\equiv\;\; & \lstinline{(x, s')}
    \end{align*}
    Thus, \lstinline{return} is not backwards round tripping.

    \item (Proof that it is weak backwards tripping)
    \begin{align*}
             & \lstinline{print (return y) x} \\
      \equiv\;\; & \lstinline{(\\_ -> Just (y, "")) x} \\
      \equiv\;\; & \lstinline{Just (y, "")}
    \end{align*}
   Thus \lstinline{s == ""}. Now we must prove the consequent of
   weak backwards round tripping:
    \begin{align*}
            & \lstinline{parse (return y) ("" ++ s')} \\
      \equiv\;\; & \lstinline{(\\s -> (y, s)) s'} \\
      \equiv\;\; & \lstinline{(y, s')} \qquad \qed
    \end{align*}
  \end{itemize}
\end{proof}

\biparserwellbehaved*

\begin{proof}
Case \lstinline{return}. Shown above.

Case \lstinline{(>>=)}.
\begin{edoc}
let (sp, v) = print p u
    (sk, w) = print (k v) u
print (p >>= k) u = (sp ++ sk, w)         -- by definition
parse p (sp ++ sk ++ s') = (v, sk ++ s')  -- by weak round tripping of p
parse (k v) sk = (w, s')                  -- by weak round tripping of k
\end{edoc}

Case \lstinline{comap}: trivial.
\end{proof}

\weakToBackward*

\begin{proof}
  The
  definition of \lstinline{purify p x = fmap fst (print p x)}
  when combined with the property \lstinline{purify p x = Just x},
  and the antecedent of backward round
  tripping (\lstinline{fmap snd (print p x) = Just s}), imply that
  \lstinline{print p x = Just (x, s)}.
  This satisfies the antecedent of weak backward round tripping, thus
  we can conclude \lstinline{parse p (s ++ s') = (x, s')},
  and thus backward round tripping holds for \lstinline{p}.
\end{proof}

In Section~\ref{subsec:compositionality} in the proof of
Corollary~\ref{cor:string-backward-rt}, we used three intermediate results about
\lstinline{char}, \lstinline{int} and \lstinline{replicateP}, namely:
\begin{align}
    \lstinline{proj char n} & \equiv \lstinline{Just n} \tag{\ref{eq:proj-char}} \\[-0.15em]
    \lstinline{proj int n} & \equiv \lstinline{Just n} \tag{\ref{eq:proj-int}} \\[-0.15em]
    \lstinline{proj (replicateBiparser (length xs) p) xs} & \equiv \lstinline{mapM (proj p) xs}
\tag{\ref{eq:proj-replicate}}
\end{align}
The first two are straightforward from their definitions.
%
 Let us take a closer look at the latter, equation~\ref{eq:proj-replicate}.

As a printer, \lstinline{replicateBiparser (length xs)} applies the printer
\lstinline{p} to every element of the input list \lstinline{xs}, and if
we ignore the output string with \lstinline{proj}, that yields
\lstinline{mapM (proj p) xs}.
When \lstinline{p} is aligned and has \lstinline{proj p = Just},
as was the case in the proof of Corollary~\ref{cor:string-backward-rt} then
all applications in the list succeed and return a \lstinline{Just}
value, so \lstinline{mapM (proj p) xs} as a whole succeeds and returns the
whole list of results. Therefore,
\lstinline{replicateBiparser (length xs) p xs = Just xs}.

\wfquasi*

\begin{proof}
  (Provides some additional details from that shown in the main body
  of the paper).
  We sketch the \ref{eqn:qcomp-bind} case, where \lstinline{p = (m >>= k)}
  for some \lstinline{m} and \lstinline{k} that are weak forward
  roundtripping.

  From \lstinline{parse (m >>= k) s01 = (y, s1)}, it follows that there exists
  \lstinline{z}, \lstinline{s} such that \lstinline{parse m s01 = (z, s)} and
  \lstinline{parse (k z) s = (y, s1)}. Similarly, the second
  conjunct \lstinline{print (m >>= k) x = Just (y, s0)}
  implies there exists \lstinline {z'}, \lstinline{s0'} such that
  \lstinline{print m x} \lstinline{=} \lstinline{Just (z', s0')} and
  \lstinline{print (k z') x} \lstinline{=} \lstinline{Just (y, s1')}
  and \lstinline{s0} \lstinline{=} \lstinline{s0' ++ s1'}.

  Because \lstinline{k} is an injective arrow,
  we have \lstinline{z = z'}. This comes from the following:
  Recall we have:
  \lstinline{parse (k z) s = (y, s1)}
  and \lstinline{print (k z') x = Just (y, s1')}
  Let \lstinline{k'} be the left arrow inverse of \lstinline{k}, then
  (from Definition~\ref{def:injective-arrow}):
  \begin{equation*}
    \lstinline{(k z >>= \y -> return (z, y)) = (k z >>= \y -> return (k' y, y))}
  \end{equation*}
(and similarly with \texttt{z'}). Plugging both sides into
``\lstinline{parse _ s}'' and ``\lstinline{print _ x}'' respectively,
then using the fact that print and parse are monad morphisms, followed
by the above two equations, we get:
\begin{align*}
& \lstinline{parse (...) s = ((z, y), s1) = ((k' y, y), s1)} \\
& \lstinline{print (...) x = Just ((z', y), s1') = Just ((k' y, y), s1')}
\end{align*}
So \lstinline{z = k' y = z'}.

  We then use the assumption that
  \lstinline{m} and \lstinline{k} are weak forward roundtripping on \lstinline{m}
  and on \lstinline{k a},
  and deduce that \lstinline{s01 = s0' ++ s} and \lstinline{s = s1' ++ s1}.
  Therefore \lstinline{s01 = s0' ++ (s1' ++ s1)} which reassociates to
  \lstinline{s01 = (s0' ++ s1') ++ s1} which equals \lstinline{s01 = s0 ++ s1}
 (since \lstinline{s0 = s0' ++ s1'}, shown above).
\end{proof}

\section{For comparison: separately defined parser and printer and round-tripping proof}
\label{app:manual-parser-printer}

The following provides a comparison between the monadic profunctor
biparser of Section~\ref{sec:biparsers} and the alternative without
our approach: having to write two separate definitions of a parser and
a printer. With these examples we also compare our reasoning approach of
Section~\ref{sec:compositionality} with having to manually prove
roundtripping on the separated definitions.

The main points of the comparison are:
\begin{itemize}
  \item This parser has the same structure as the biparser in Section~\ref{sec:biparsers},
    only without any \lstinline{upon} annotations.
  \item The standalone printer for this example is extremely simple.
  \item The auxiliary lemmas for
    \lstinline{parseInt}, \lstinline{parseDigits}, and \lstinline{replicateM}
    correspond to the round tripping properties of their bidirectional
    counterparts (\lstinline{int}, \lstinline{digits}, and
    \lstinline{replicateBiparser}).
  \item These lemmas must all manipulate source strings explicitly,
    whereas our framework uses compositionality to handle those details.
\end{itemize}

\subsection{Parser monad}

\begin{code}
newtype Parser a = Parser { runParser :: String -> (a, String) }

instance Monad Parser where
  return x = Parser (\s -> (x, s))
  p >>= k = Parser (\s ->
    let (y, s') = runParser p s in
    runParser (k y) s')

parseChar :: Parser Char
parseChar = Parser (\(c : s) -> (c, s))
\end{code}

\subsection{String parser}

\begin{code}
parseDigits :: Parser String
parseDigits = do
  d <- parseChar
  if isDigit d then do
    igits <- parseDigits
    return (d : igits)
  else if d == ' ' then
    return ""
  else
    error "Expected digit or space."

parseInt :: Parser Int
parseInt = do
  ds <- parseDigits
  return (readInt ds)

parseString :: Parser String
parseString = do
  n <- parseInt
  replicateM n parseChar

replicateM :: Int -> Parser a -> Parser [a]
replicateM 0 _ = return []
replicateM n p = do
  x <- p
  xs <- replicateM (n - 1) p
  return (x : xs)
\end{code}

\subsection{String printer}

\begin{code}
type Printer a = a -> String

printString :: Printer String
printString s = showInt (length s) ++ " " ++ s
\end{code}

\subsection{Backwards round tripping}

\begin{theorem}
  For any string \lstinline{s}:
  \lstinline{runParser parseString (printString s) = (s, "")}.
\end{theorem}

\begin{proof} By equational reasoning.
\begin{edoc}
runParser (parseInt >>= \n -> replicateM n parseChar)
          (showInt (length s) ++ " " ++ s)
    {- Lemma (parseInt) -}
  = runParser (replicateM (length s) parseChar) s
    {- Lemma (replicateM) -}
  = (s, "")
\end{edoc}
\end{proof}

\subsubsection{Auxiliary lemmas}

\begin{nlem}[{\texttt{\small parseInt}}] \normalfont
  For any nonnegative integer \lstinline{n}, and any string \lstinline{s},
  \begin{equation*}
    \lstinline{runParser parseInt (showInt n ++ " " ++ s) = (n, s)}
  \end{equation*}
\end{nlem}

\begin{proof} By equational reasoning.
\begin{edoc}
runParser (parseDigits >>= \ds -> return (readInt ds))
          (showInt n ++ " " ++ s)
    {- Lemma (parseDigits), assuming (showInt n) is a string of digits -}
  = (readInt (showInt n), s)
    {- Assuming readInt is a one-sided inverse of showInt -}
  = (n, s)
\end{edoc}
\end{proof}

\begin{nlem}[{\texttt{\small parseDigits}}] \normalfont
    For any string of digits \lstinline{ds}, and any string \lstinline{s},
  \begin{equation*}
    \lstinline{runParser parseDigits (ds ++ " " ++ s) = (ds, s)}
  \end{equation*}
\end{nlem}

\begin{proof} By equational reasoning, and by induction on \lstinline{ds}.
\begin{itemize}
  \item Case \lstinline{ds = ""}:
\begin{edoc}
runParser (parseChar >>= \d -> if d ...) (" " ++ s)
    {- Definition of parseChar, d = ' ' -}
  = runParser (return "") s
  = ("", s)
\end{edoc}
  \item Case \lstinline{ds = (d1:ds1)}, \lstinline{d1} is a digit:
\begin{edoc}
runParser (parseChar >>= \d -> if d ...) (d1 : ds1 ++ " " ++ s)
    {- Definition of parseChar, d = d1 -}
  = runParser (parseDigits >>= \igits ->
                 return (d1 : igits)) (ds1 ++ " " ++ s)
    {- Induction hypothesis on ds1 -}
  = (d1 : ds1, s)
\end{edoc}
\end{itemize}
\end{proof}

\begin{nlem}[{\texttt{\small replicateM}}] \normalfont
  For any string \lstinline{s},
  \begin{equation*}
    \lstinline{runParser (replicateM (length s) parseChar) s = (s, "")}
  \end{equation*}
\end{nlem}

\begin{proof} By equational reasoning, and by induction on \lstinline{ds}.
\begin{itemize}
  \item Case \lstinline{s = ""}:
\begin{edoc}
runParser (return "") "" = ("", "")
\end{edoc}
  \item Case \lstinline{s = (c1:s1)}:
\begin{edoc}
runParser (char >>= \c -> replicateM (length s1) parseChar >>= \s ->
             return (c : s))
          (c1 : s1)
    {- Definition of parseChar -}
  = runParser (replicateM (length s1) parseChar >>= \s ->
                 return (c1 : s))
    {- Induction hypothesis on s1 -}
  = runParser (return (c1 : s1)) ""
  = (c1 : s1, "")
\end{edoc}
\end{itemize}
\end{proof}

\section{Lenses}

\thmWbRt*

\begin{proof}
Assume the antecedent of backward roundtripping:
\begin{equation*}
 \lstinline{put l x s = Just ((y, s'), p')} \wedge \lstinline{p' s' = True}
\end{equation*}
The goal is to prove \lstinline{get l s' = Just x}.

By the identity projection premise we have that
\lstinline{proj s l x = Just x}
for all \lstinline{s}. Recall the definition of \lstinline{proj} for lenses:
\begin{equation*}
\lstinline{proj s l = \u -> fmap (fst . fst) (put l u s)}
\end{equation*}
Combining this with assumption on \lstinline{put} and identity project
we see that:
$$
\lstinline{put l x s = Just ((x, s'), p')}
$$
We can thus instantiate weak backward round tripping to get the desired goal:
\begin{equation*}
\lstinline{get l s' = Just x}
\end{equation*}
\end{proof}

\atkeyWb*

\begin{proof}
  Recall \lstinline{atKey :: L Src Value Value}.
  Assuming the antecedent of backward round-tripping, we get the following information:
\begin{gather*}
  \lstinline{put l x (atKey k) m}
   = \lstinline{Just ((x, insert k x m), \\m' -> lookup k m' == Just
     x)} \\
\lstinline{(\\m' -> lookup k m' == Just x) s'} = \lstinline{True}
\end{gather*}
We then need to prove \lstinline{get l s' = Just x}. By the
definition of \lstinline{get}:
\begin{edoc}
  get (atKey k) s' = lookup k s'
\end{edoc}
By the second conjunct of the antecedent we know \lstinline{lookup k s' = Just x},
giving the required consequent.
\end{proof}

\atkeyPp*

\begin{proof}
For all $\lstinline{s :: s}$, following the definition we get:
\begin{align*}
& \lstinline{proj s (atKey k)} \\
\equiv \; & \lstinline{\\u -> fmap (fst . fst) (put (atKey k) u s)} \\
\equiv \; & \lstinline{\\u -> fmap (fst . fst) (Just ((u, ...), ...))} \\
\equiv \; & \lstinline{\\u -> Just u}
\end{align*}
\end{proof}

\subsection{Further example: Lenses Over Trees}

\noindent
Our lens structuring provides the following two smart deconstructors
and one smart constructor:
\begin{edoc}
get :: L s u v -> (s -> Maybe v)
put :: L s u v -> (u -> s -> Maybe ((v, s), s -> Bool))
mkLens :: (s -> Maybe v) -> (u -> s -> Maybe ((v, s), s -> Bool)) -> L s u v
\end{edoc}

As an example of programming with monadic lenses, we consider lenses over
the following data type of binary trees labeled by integers.
\begin{code}
data Tree = Leaf | Node Tree Int Tree deriving Eq
\end{code}
%
%
In this example, our aim is to build a lens whose forward direction
gets the right spine of the tree as a list of integers. The backwards
direction will then allow a tree to be updated with a new right spine
(represent as a list of integers), which may produce a larger source
tree.

We start by defining the classical lens combinator. Given a lens \lstinline{lt} to view \lstinline{s} as \lstinline{t},
and a lens \lstinline{ly} to view \lstinline{t} as \lstinline{u},
the combinator \lstinline{(>>>)} creates a lens to view
\lstinline{s} as \lstinline{u}. We illustrate and explain
the composition on the right.

\vspace{0.5em}
\hspace{-1.2em}
\begin{minipage}{0.53\linewidth}
\begin{code}
(>>>) :: L s t t -> L t u u -> L s u u
lt >>> ly = mkLens get' put' where
  -- get' :: s -> Maybe u
  get' s = get lt s >>= get ly

  -- put' :: u -> s ->
  --        Maybe ((u, s), s -> Bool)
  put' xu s =
   case get lt s of
     Nothing -> Nothing
     Just t -> do
        ((y, xt), q') <- put ly xu t
        ((_, s'), p') <- put lt xt s
        if q' xt
          then Just ((y, s'), p')
          else Nothing




\end{code}
\end{minipage}
\vline{}\vspace{0.5em}
\hspace{1em}
\begin{minipage}{0.39\linewidth}
Illustration of the composition of lenses in
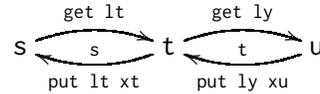
\lstinline{(>>>)}:
\vspace{0.5em}
%
\begin{align*}
\scalebox{1.1}{
\xymatrix@C=4em{
\texttt{s} \ar@/^/[r]^{\texttt{get lt}}
& \texttt{t} \ar@/^/[l]^{\texttt{put lt xt}}_{\texttt{s}} \ar@/^/[r]^{\texttt{get ly}}
& \texttt{u} \ar@/^/[l]^{\texttt{put ly xu}}_{\texttt{t}}
}
}
\vspace{0.5em}
\end{align*}
For individual lenses,
the \lstinline{put} action takes the source
as its last parameter (shown above the lower
arrows here). In the case of the composite
lens, \lstinline{put'} has a source of type
\lstinline{s}, thus we need to create an intermediate
source of type \lstinline{t} in order to use
\texttt{put ly}. This intermediate source is provided
by first using \texttt{get lt s}.
\end{minipage}\vspace{0.5em}

\noindent
In the last three lines of \lstinline{putter} in
\lstinline{(>>>)}, in order for the composite backwards direction to
succeed, the returned intermediate store \lstinline{xt} must be
consistent (free of conflict) as checked by \lstinline{q' xt}.

We define two primitive lenses:
\lstinline{rootL} and \lstinline{rightL} for the root and
right child of a tree:
%

\noindent
\hspace{-1em}\begin{minipage}[t]{0.54\linewidth}
\begin{code}
rootL :: L Tree (Maybe Int) (Maybe Int)
rootL = mkLens getter putter
  where
   getter (Node _ n _) = Just (Just n)
   getter Leaf         = Just Nothing

   putter n' t = Just ((n', t'), p)
    where
      t' = case (t, n') of
        (_,    Nothing) -> Leaf
        (Leaf, Just n)  -> Node Leaf n Leaf
        (Node l _ r, Just n)  -> Node l n r
      p t'' = getter t' == getter t''
\end{code}
\end{minipage}
\begin{minipage}[t]{0.4\linewidth}
\begin{code}
rightL :: L Tree Tree Tree
rightL = mkLens getter putter
  where
   getter (Node _ _ r) = Just r
   getter  _           = Nothing

   putter r Leaf = Nothing
   putter r (Node l n _) = Just
     ((r, Node l n r),
       \t' -> Just r == getter t')
\end{code}
\end{minipage}
\vspace{0.25em}

\noindent
The \lstinline{rootL} lens accesses the label at the root if it is a \lstinline{Node},
otherwise returning \lstinline{Nothing}. Note that \lstinline{Maybe}
type used here is different to use the of \lstinline{Maybe} inside the
definition of \lstinline{L}: internally \lstinline{L} uses
\lstinline{Maybe} to represent failure, here at the top-level we are
using it to merely indicate presence of absence of a label.

The second lens \lstinline{rightL}
accesses the right child of a tree which can however fail if the source tree
is a \lstinline{Leaf} rather than a \lstinline{Node}.

Both lenses provide \lstinline{put} operations which return
predicates that check that the view of a store is equal to the view
of the store updated by the put.

We compose these two primitive lenses monadically to define the
\lstinline{spineL} lens to view and update the right spine
of a tree:

\begin{code}
spineL :: L Tree [Int] [Int]
spineL = do
  hd <- comap (Just . safeHead) rootL
  case hd of
    Nothing -> return []
    Just n -> do
      tl <- comap safeTail (rightL >>> spineL)
      return (n : tl)
\end{code}
\noindent Auxiliary functions \lstinline{safeHead} and \lstinline{safeTail} are defined:
\begin{code}
safeHead :: [a] -> Maybe a
safeHead (a : _) = Just a
safeHead [] = Nothing

safeTail :: [a] -> Maybe [a]
safeTail (_ : as) = Just as
safeTail [] = Nothing
\end{code}

\noindent
As a \lstinline{get}, it first views the root of the source tree
through \lstinline{rootL} as \lstinline{hd}, and whether it recurse or
not depends on whether it is a node (with label \lstinline{n}) or a
leaf, using \lstinline{rightL} to shift the context.  As a
\lstinline{put}, it updates the root using the head of the list, which
is returned as the view \lstinline{hd}, and continues with the same
logic. To illustrate the action of this lens, consider a tree:
\begin{code}
t0 = Node (Node Leaf 0 Leaf) 1 (Node Leaf 2 Leaf)
\end{code}
Getting the right spine \,(\lstinline{get spineL t0})\, yields the list
\lstinline{[1, 2]}. The tree spine can be updated to
\lstinline{[3, 4, 5]} yielding the following tree:
\begin{edoc}
fmap fst (put spineL [3, 4, 5] t0)
  = Just ([3, 4, 5], Node (Node Leaf 0 Leaf) 3 (Node Leaf 4 (Node Leaf 5 Leaf)))
\end{edoc}


\section{Generators}\label{sec:generator_example}

This appendix section provides additional code examples for our notion
of \emph{bigenerators} (\emph{bidirectional generators}) extending
random generators in property-based testing frameworks like
\emph{QuickCheck}~\citep{Claessen2000} to a bidirectional setting.

We assume given a \lstinline{Gen} monad of random generators (e.g. as
defined in the QuickCheck library for Haskell) and two
primitive generators: \lstinline{genBool ::} \lstinline{Double -> Gen Bool}
generates a random boolean according to a Bernoulli distribution with a given
parameter $p\in[0,1]$; \lstinline{choose :: (Int, Int) -> Gen Int} generates a
random integer uniformly in a given inclusive range $[\mathrm{min},
\mathrm{max}]$.

\paragraph{Generators for binary search trees}
We consider again the type of trees from the previous section. A \emph{binary
search tree} (BST) is a \lstinline{Tree} whose nodes are in sorted order.
Inductively, a BST is either a \lstinline{Leaf}, or some \lstinline{Node l n r}
where \lstinline{l} and \lstinline{r} are both binary search trees, nodes in
\lstinline{l} have smaller values than \lstinline{n}, and nodes in
\lstinline{r} have greater values than \lstinline{n}.

As a working example, we are given some function
\lstinline{insert :: Tree -> Int -> Tree} which inserts an integer in a BST. We
want to test the invariant that BSTs are mapped to BSTs, by \emph{generating} a
BST and an integer to apply the \lstinline{insert} function, and \emph{check}
that the output is also a BST.

With the \lstinline{Gen} monad, we can write a simple generator of BSTs
recursively: given some bounds on the values of the nodes,
if the bounds describe a nonempty interval, we flip a coin to decide whether to
generate a leaf or a node, and if it is a node, we recursively generate binary
search trees, following the inductive definition above.
We can similarly write a \emph{checker} for binary search trees as
a predicate.
\begin{minipage}{0.5\linewidth}
\begin{edoc}
genBST :: Int -> Int -> Gen Tree
genBST min max | max < min = return Leaf
genBST min max = do
  isLeaf' <- genBool 0.5
  if isLeaf' then return Leaf
  else do n <- choose (min, max)
          l <- genBST min (n-1)
          r <- genBST (n+1) max
          return (Node l n r)
\end{edoc}
\end{minipage}
\hspace{1em}
\begin{minipage}{0.5\linewidth}
\begin{edoc}
checkBST :: Int -> Int -> Tree -> Bool
checkBST min max Leaf = True
checkBST min max (Node l n r) =
     min <= n
  && n <= max
  && checkBST l
  && checkBST r
\end{edoc}
\vspace{2.7em}
\end{minipage}

%

\paragraph{Bigenerator}

A generator of values \texttt{v} and a predicate on \texttt{v}
(modelled by \lstinline{v -> Bool}) together define a bidirectional generator
with the same pre-view and view type, provided here by a smart
constructor: \lstinline{mkAlignedG}:
\begin{code}
mkAlignedG :: Gen v -> (v -> Bool) -> G v v
mkAlignedG gen check = mkG gen (\y -> if check y then Just y else Nothing)
\end{code}
Recall from Section~\ref{sec:generators} that
a bigenerator can be mapped to a predicate via \lstinline{toPredicate}:
\begin{code}
toPredicate :: G u v -> u -> Bool
toPredicate g x = isJust (check g x) where
  isJust (Just _) = True
  isJust Nothing  = False
\end{code}

\noindent
We wrap two generator primitives as \lstinline{bool} and \lstinline{inRange}.
As predicates, \lstinline{bool} makes no assertion, \lstinline{inRange}
checks that the input integer is within the given range.

\begin{minipage}{0.45\linewidth}
\begin{code}
bool :: Double -> G Bool Bool
bool p = mkAlignedG
  (genBool p)
  (\_ -> True)
\end{code}
\end{minipage}
\begin{minipage}{0.45\linewidth}
\begin{code}
inRange :: (Int, Int) -> G Int Int
inRange (min, max) = mkAlignedG
  (choose (min, max))
  (\x -> min <= x && x <= max)
\end{code}
\end{minipage}
\vspace{0.5em}

\noindent
We consider again a type of labelled trees,
with some field accessors. On the bottom right,
\lstinline{leaf} is a simple bigenerator for leaves.
\begin{code}
data Tree = Leaf | Node Tree Int Tree

nodeValue :: Tree -> Maybe Int
nodeValue (Node _ n _) = Just n
nodeValue _            = Nothing
\end{code}

\noindent
\begin{minipage}[b]{0.5\linewidth}
\begin{code}
nodeLeft, nodeRight :: Tree -> Maybe Tree
nodeLeft (Node l _ _) = Just l
nodeLeft _            = Nothing

nodeRight (Node _ _ r) = Just r
nodeRight _            = Nothing
\end{code}
\end{minipage}
\hspace{0.5em}
\begin{minipage}[b]{0.4\linewidth}
\begin{code}
isLeaf :: Tree -> Bool
isLeaf Leaf = True
isLeaf (Node _ _ _) = False

leaf :: G Tree Tree
leaf = mkAlignedG (return Leaf) isLeaf
\end{code}
\end{minipage}

\noindent
We then define a specification of binary search trees (\lstinline{bst} below).
A corresponding generator and predicate are extracted on the right from this
bigenerator:

\hspace{-2em}
\begin{minipage}[t]{0.6\linewidth}
\begin{code}
bst :: (Int, Int) -> G Tree Tree
bst (min, max) | min > max = leaf
bst (min, max) = do
  isLeaf' <- comap (Just . isLeaf) (bool 0.5)
  if isLeaf' then return Leaf
  else do
    n <- comap nodeValue (inRange (min, max))
    l <- comap nodeLeft  (bst (min, n - 1))
    r <- comap nodeRight (bst (n + 1, max))
    return (Node l n r)
\end{code}
\end{minipage}
%
\hspace{0.5em}
\begin{minipage}[t]{0.45\linewidth}
\begin{code}
genBST :: Gen Tree
genBST =
  generate (bst (0, 20))

checkBST :: Tree -> Bool
checkBST =
  toPredicate (bst (0, 20))
\end{code}
\end{minipage}

\noindent
As a predicate, \lstinline{bst} first checks whether the root is a leaf (\lstinline{isLeaf});
returning a boolean allows us to reuse the same case expression as for the
generator. If it is a node, we check that the value is within the given
range and then recursively check the subtrees.

%

\fi

\end{document}